\documentclass[a4paper,UKenglish]{lipics-v2019}

\usepackage{macros}
\makeatletter
\if@todonotes@disabled

\else

\fi
\makeatother

\title{Parameterised Complexity of Abduction in Schaefer's Framework} %TODO Please add

\author{Yasir Mahmood}{Leibniz Universität Hannover, Institut für Theoretische Informatik}{mahmood@thi.uni-hannover.de}{https://orcid.org/0000-0002-5651-5391}{Funded by the German Research Foundation (DFG) project ME 4279/1-2}

\author{Arne Meier}{Leibniz Universität Hannover, Institut für Theoretische Informatik}{meier@thi.uni-hannover.de}{https://orcid.org/0000-0002-8061-5376}{Funded partially by the German Research Foundation (DFG) project ME 4279/1-2}

\author{Johannes Schmidt}{Department of Computer Science and Informatics, School of Engineering, Jönköping University}{johannes.schmidt@ju.se}{https://orcid.org/0000-0001-8551-1624}{}

\authorrunning{Y.~Mahmood, A.~Meier, and J.~Schmidt}
\Copyright{Yasir Mahmood, Arne Meier, and Johannes Schmidt}

%TODO mandatory; please add comma-separated list of keywords

\nolinenumbers %uncomment to disable line numbering

\hideLIPIcs  %uncomment to remove references to LIPIcs series (logo, DOI, ...), e.g. when preparing a pre-final version to be uploaded to arXiv or another public repository

\begin{document}

\maketitle
\begin{abstract}
	Abductive reasoning is a non-monotonic formalism stemming from the work of Peirce.
	It describes the process of deriving the most plausible explanations of known facts.
	Considering the positive version asking for sets of variables as explanations, we study, besides asking for existence of the set of explanations, two explanation size limited variants of this reasoning problem (less than or equal to, and equal to).
	In this paper, we present a thorough two-dimensional classification of these problems.
	The first dimension is regarding the parameterised complexity under a wealth of different parameterisations.
	The second dimension spans through all possible Boolean fragments of these problems in Schaefer's constraint satisfaction framework with co-clones~(STOC~1978).
	Thereby, we almost complete the parameterised picture started by Fellows~et~al.~(AAAI~2012), partially building on results of Nordh and Zanuttini~(Artif.~Intell.~2008).
	In this process, we outline a fine-grained analysis of the inherent parameterised  intractability of these problems and pinpoint their FPT parts. As the standard algebraic approach is not applicable to our problems, we develop an alternative method that makes the algebraic tools partially available again.
\end{abstract}

\section{Introduction}
The framework of parameterised complexity theory yields a more fine-grained complexity analysis of problems than classical worst-case complexity may achieve.
Introduced by Downey and Fellows \cite{DBLP:series/txcs/DowneyF13,DBLP:series/mcs/DowneyF99}, one associates problems with a specific \emph{parameterisation}, that is, one studies the complexity of \emph{parameterised problems}.
Here, one aims to find parameters relevant for practice allowing to solve the problem by algorithms running in time $f(k)\cdot n^{O(1)}$, where $f$ is a computable function, $k$ is the value of the parameter and $n$ is the input length.
Problems with such a running time are called \emph{fixed-parameter tractable} (\FPT) and correspond to efficient computation in the parameterised setting.
This is justified by the fact that parameters are usually slowly growing or even of constant value.
Despite that, a different quality of runtimes is of the form $n^{f(k)}$ which are obeyed by algorithms solving problems in the class $\XP$.
Comparing both classes with respect to the runtimes their problems allow to be solved in, of course, both runtimes are polynomial.
However, for the first type, the degree of the polynomial is independent of the parameter's value which is notable to observe.
As a result, the second kind of runtimes is undesirable and usually tried to circumvented by locating different parameters.
It is known that $\FPT\subsetneq\XP$ by diagonalisation and also that a (presumably infinite) hierarchy of parameterised intractability in between these two classes exist: the so-called $\complClFont{W}$-hierarchy which is contained also in the class $\W\P\subseteq\XP$.
These $\complClFont{W}$-classes are regarded as a measure of intractability in the parameterised sense.
Intuitively, showing $\W1$-lower bounds corresponds to $\NP$-lower bounds in the classical setting.
The limit of this hierarchy, the class $\W\P$ is defined via nondeterministic machines that have at most $h(k)\cdot\log n$ many nondeterministic steps, where $h$ is a computable function, $k$ the parameter's value, and $n$ is the input length.
Clearly, human common-sense reasoning is a non-monotonic process as adding further knowledge might decrease the number of deducible facts.
As a result, non-monotonic logics became a well-established approach to investigate this kind of reasoning.
One of the popular formalism in this area of research is abductive reasoning which is an important concept in artificial intelligence as emphasised by Morgan \cite{DBLP:journals/ai/Morgan71} and Pole \cite{DBLP:conf/ijcai/Poole89}. 
In particular, abduction is used in the process of medical diagnosis \cite{pr90,DBLP:journals/tsmc/JosephsonCST87} and thereby relevant for practice.
Intuitively, abductive reasoning describes the process of deriving the most plausible explanations of known facts and originated from the work of Peirce \cite{peirce}.
Formally, one uses propositional formulas to model known facts in a \emph{knowledge base} $\KB$ together with a set of \emph{manifestations} $M$ and a set of \emph{hypotheses} $H$.
In this paper, $H$ and $M$ are sets of propositions as studied by Fellows~et~al.~\cite{DBLP:conf/aaai/FellowsPRR12} as well as Eiter and Gottlob~\cite{DBLP:journals/jacm/EiterG95}.
Formally, one tries to find a preferably small set of propositions $E\subseteq H$ such that $E\land\KB$ is satisfiable and $E\land \KB\models M$. 
$E$ is then called an explanation for $M$.
In this context, we distinguish three kinds of problems: the first just asks for such a very set $E$ that fulfils these properties ($\problemFont{ABD}$), the second tries to find a set of size less than or equal to a specific size ($\problemFont{ABD}_\leq$), and the third one wants to spot a set of exactly a given size ($\problemFont{ABD}_=$).
Classically, $\problemFont{ABD}$ is complete for the second level of the polynomial hierarchy $\SigmaP$ \cite{DBLP:journals/jacm/EiterG95} and its difficulty is very well understood~\cite{DBLP:conf/aaai/SelmanL90,DBLP:journals/siamcomp/CreignouZ06,DBLP:journals/ai/NordhZ08,DBLP:conf/kr/CreignouST10}.
As a result, under reasonable complexity-theoretic assumptions, the problem is highly intractable posing the question in turn for sources of this complexity.
In this direction, there exists research that aims to better understand the structure and difficulty of this problem, namely, in the context of parameterised complexity.
Here, Fellows et~al.~\cite{DBLP:conf/aaai/FellowsPRR12} initiated an investigation of possible parameters and classified CNF-induced fragments of the reasoning problems with respect to a multitude of parameters.
The authors study CNF-fragments with respect to the classes \textsc{Horn}, \textsc{Krom}, and \textsc{DefHorn}.
They studied the parameterisations $|M|$ (number of manifestations), $|H|$ (number of hypotheses), $|V|$ (number of variables), $|E|$ (number of explanations which is equivalent to their solution size $k$) directly stemming from problem components, as well as the tree-width \cite{DBLP:journals/jal/RobertsonS86}, and the size of the smallest vertex cover.
In their classification, besides showing several $\para\NP$-/$\W\P$-complete/$\FPT$ cases, they also focus on the existence of polynomial kernels and present a complete picture regarding their CNF-classes.

Universal algebra yields a systematic way to rigorously classify fragments of a problem induced by restricting its Boolean connectives.
This technique is built around Post's lattice~\cite{pos41} which bases on the notion of (co-)clones.
Intuitively, given a set of Boolean functions $B$, the \emph{clone} of $B$ is the set of functions that are expressible by compositions of functions from $B$ (plus introducing fictive variables).
The most prominent result under this approach is the dichotomy theorem of Lewis \cite{lew79} which classifies propositional satisfiability into polynomial-time solvable cases and intractable ones depending merely on the existence of specific Boolean operators.
This approach has been followed many times in a wealth of different contexts \cite{DBLP:journals/corr/abs-0812-4848,DBLP:journals/logcom/BeyersdorffMTV12,cmtv12,DBLP:journals/argcom/Creignou0TW11,DBLP:journals/tcs/Meier013,DBLP:journals/ijfcs/MeierTVM09,DBLP:phd/de/Reith2002} as well as in the context of abduction itself \cite{DBLP:journals/ai/NordhZ08,DBLP:journals/logcom/CreignouST12}.
Interestingly, in the scope of constraint satisfaction problems, the investigation of co-clones (or relational clones) allows one to proceed a similar kind of classification (see, e.g., the work of Nordh and Zanuttini~\cite{DBLP:journals/ai/NordhZ08}).
The reason for that lies in the concept of invariance of relations under some function $f$ (one defines this property via \emph{polymorphisms} where $f$ is applied component-wisely to the columns of the relation).
In view of this, Post's lattice supplies a similar lattice, now for sets of relations which are invariant under respective functions.
With respect to constraint satisfaction, the most prominent classification is due to Schaefer \cite{DBLP:conf/stoc/Schaefer78} who similarly divides the constraint satisfaction problem restricted to co-clones into polynomial-time solvable and $\NP$-complete cases.
The algebraic approach has been successfully applied to abduction by Nord and Zanuttini~\cite{DBLP:journals/ai/NordhZ08}. 
For the problems that we consider, it is less obvious how to use the algebraic tools: the standard trick to obtain reductions preserves existence of explanations, but not their size. 
Due to this, we develop an alternative method that makes the algebraic tools partially available again (see Section~\ref{subsec:base-independence}).

Much in the vein of Schaefer's classification, we present a thorough study directly pinpointing those restrictions of the abductive reasoning problem which yield efficiency under the parameterised approach.
In a sense, we present an almost complete picture which has been initiated by Fellows et~al.~\cite{DBLP:conf/aaai/FellowsPRR12} except for some minor cases around the affine co-clones.
Their classification is covered by our study now, as \textsc{Horn} cases correspond to the co-clones below $\IE_2$, \textsc{DefHorn} conforms $\IE_1$, and \textsc{Krom} matches with $\ID_2$.
The motivation of our research is to draw a finer line than Fellow et~al.\ did and to present a completer picture with respect to all possible constraint languages now.
From this classification, we draw some surprising results.
Regarding the essentially negative cases for the parameter $|M|$, $\problemFont{ABD}_=$ is $\para\NP$-complete whereas $\problemFont{ABD}_\leq$ is $\FPT$. 
Also for this parameter, $\IE_1$ and $\IE$ are hard for $\problemFont{ABD}_=$ and $\problemFont{ABD}_\leq$ (both $\para\NP$-complete) but $\problemFont{ABD}$ is $\FPT$.
Regarding $|E|$ as parameterisation, the behaviour is similarly unexpected for the essentially negative cases: $\FPT$ for $\problemFont{ABD}_\leq$ versus $\W1$-hardness for $\problemFont{ABD}_=$.
For the parameters $|V|$ as well as $|H|$ the classifications for all three problems are the same. 
Figure~\ref{fig:overview} shows our results for all problems and parameterisations in a single picture.
Proof details in the appendix are symbolised by `$\star$'.

\begin{figure}
	\centering
	\includegraphics[width=\linewidth]{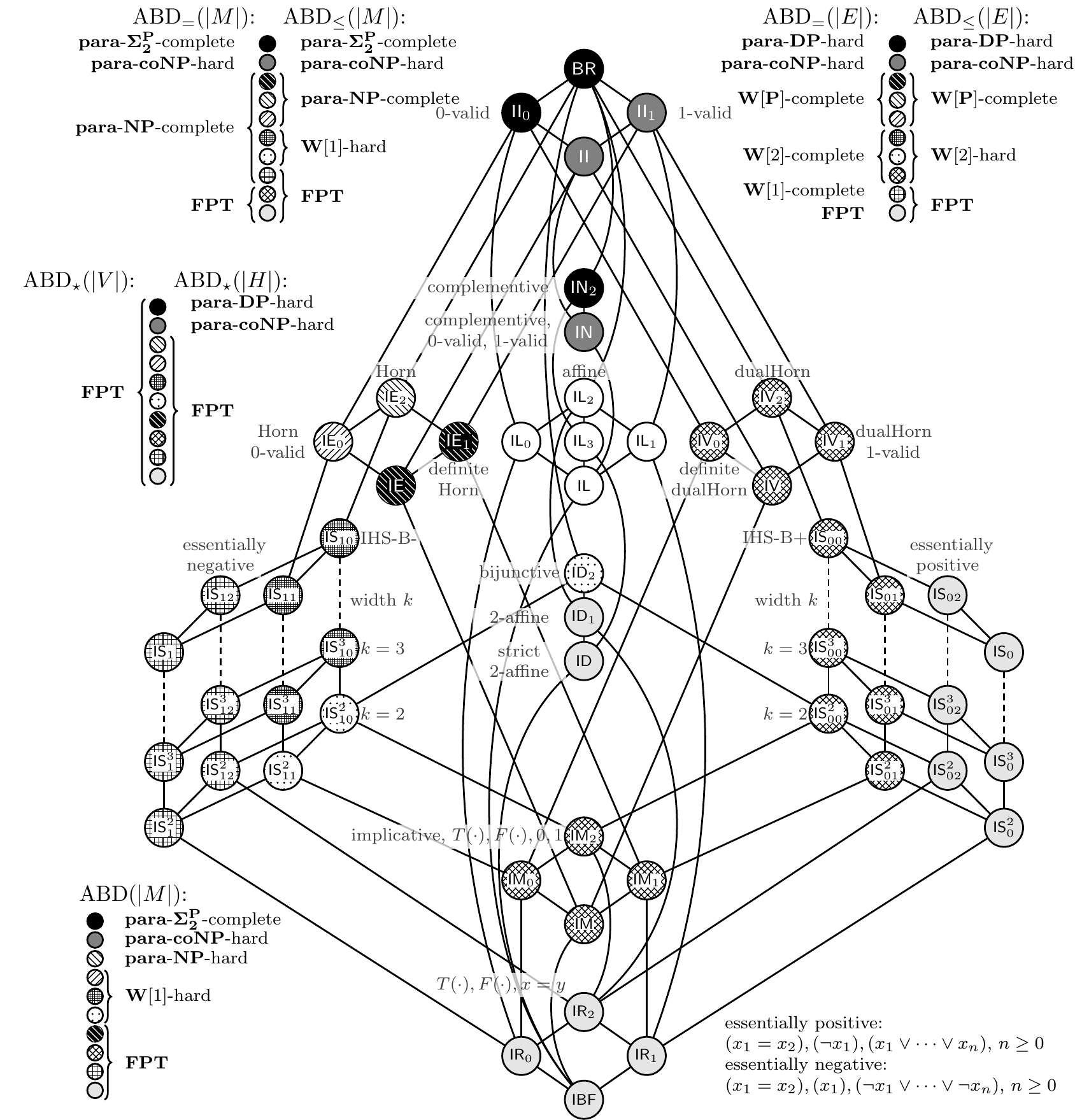}
	\caption{Complexity landscape of abductive reasoning with respect to the studied parameters $|M|,|H|,|V|,|E|$. Notice, that due to presentation reasons, some completeness results are just mentioned with their lower bound, e.g., case $\abdle{\IS{2}{11}}{|E|}$ is \W2-complete (Theorem~\ref{theorem-E}).
	White colouring means unclassified. $\problemFont{ABD}_\star$ means same result for all three variants.}\label{fig:overview}
\end{figure}

\section{Preliminaries}

We require standard notions from classical complexity theory \cite{Papadimitriou94}. 
We encounter the classical complexity classes %\parityL,
\P, \NP, $\DP = \{A \setminus B \mid A,B\in \NP\}$, $\co\NP$, $\SigmaP=\NP^\NP$ and their respective completeness notions, employing polynomial time many-one reductions ($\preduction$). 

\subparagraph*{Parameterised Complexity Theory.}
A \emph{parameterised problem} (PP) $P\subseteq\Sigma^*\times\mathbb N$ is a subset of the crossproduct of an alphabet and the natural numbers.
For an instance $(x,k)\in\Sigma^*\times\mathbb N$, $k$ is called the (value of the) \emph{parameter}.
A \emph{parameterisation} is a polynomial-time computable function that maps a value from $x\in\Sigma^*$ to its corresponding $k\in\mathbb N$.
The problem $P$ is said to be \emph{fixed-parameter tractable} (or in the class $\FPT$) if there exists a deterministic algorithm $\mathcal A$ and a computable function $f$ such that for all $(x,k)\in\Sigma^*\times\mathbb N$, algorithm $\mathcal A$ correctly decides the membership of $(x,k)\in P$ and runs in time $f(k)\cdot|x|^{O(1)}$.
The problem $P$ belongs to the class $\XP$ if $\mathcal A$ runs in time $|x|^{f(k)}$.
There exists a hierarchy of complexity classes in between $\FPT$ and $\XP$ which is called $\W{}$-hierarchy (for details see the textbook of Flum and Grohe~\cite{DBLP:series/txtcs/FlumG06}). 
We will make use of the classes $\W1$ and $\W2$.
Complete problems characterising these classes are introduced later in Proposition~\ref{theorem-wsat}. 
Also, we work with classes that can be defined via a precomputation on the parameter.
\begin{definition}
	Let $\mathcal C$ be any complexity class.
	Then $\para\mathcal C$ is the class of all PPs $P\subseteq\Sigma^*\times\mathbb N$ such that there exists a computable function $\pi\colon\mathbb N\to\Delta^*$ and a language $L\in\mathcal C$ with $L\subseteq\Sigma^*\times\Delta^*$ such that for all $(x,k)\in\Sigma^*\times\mathbb N$ we have that $(x,k)\in P \Leftrightarrow (x,\pi(k))\in L$.
\end{definition}
Notice that $\para\complClFont{P}=\FPT$.
The complexity classes $\mathcal{C}\in\setdefinition{  \NP,\co\NP, \DP,\SigmaP}$ are used in the $\para\mathcal C$ context by us.

Let $c\in\mathbb N$ and $P\subseteq\Sigma^*\times\mathbb N$ be a PP, then the \emph{$c$-slice of $P$}, written as $P_c$ is defined as $P_c:=\{\,(x,k)\in\Sigma^*\times\mathbb N\mid k=c\,\}$.
Notice that $P_c$ is a classical problem then.
Observe that, regarding our studied complexity classes, showing membership of a PP $P$ in the complexity class $\para\mathcal C$, it suffices to show that each slice $P_c\in\mathcal C$.

\begin{definition}\label{def:fpt-reduction}
	Let $P\subseteq\Sigma^*\times\mathbb N,Q\subseteq\Gamma^*$ be two PPs.
	One says that $P$ is \emph{fpt-reducible} to $Q$, $P\fptreduction Q$, if there exists an fpt-computable function $f\colon\Sigma^*\times\mathbb N\to\Gamma^*\times\mathbb N$ such that
	\begin{itemize}
		\item for all $(x,k)\in\Sigma^*\times\mathbb N$ we have that $(x,k)\in P\Leftrightarrow f(x,k)\in Q$,
		\item there exists a computable function $g\colon\mathbb N\to\mathbb N$ such that for all $(x,k)\in\Sigma^*\times\mathbb N$ and $f(x,k)=(x',k')$ we have that $k'\leq g(k)$.
	\end{itemize}
\end{definition}

\subparagraph*{Propositional Logic.}
We assume familiarity with propositional logic. 
A \emph{literal} is a variable $x$ or its negation $\neg x$. 
A \emph{clause} is a disjunction of literals and a \emph{term} is a conjunction of literals. 
We denote by $\var(\varphi)$ the variables of a formula $\varphi$. 
Analogously, for a set of formulas $F$, $\var(F)$ denotes $\bigcup_{\varphi \in F}\var(\varphi)$. 
We identify finite $F$ with the conjunction of all formulas from $F$, that is, $\bigwedge_{\varphi \in F} \varphi$. 
A mapping $\sigma\colon \var(\varphi) \mapsto \{0,1\}$ is called an \emph{assignment} to the variables of $\varphi$. 
A \emph{model} of a formula $\varphi$ is an assignment to $\var(\varphi)$ that satisfies $\varphi$. 
The \emph{weight} of an assignment $\sigma$ is the number of variables $x$ such that $\sigma(x)=1$.
For two formulas $\psi, \varphi$ we write $\psi  \models \varphi$ if every model of $\psi$ also satisfies $\varphi$. 
A formula is positive (resp.\ negative) if every literal appears positively (negatively) and a negation symbol appears only in front of a variable. 
The class of all propositional formulas is denoted by $\PROP$.
Occasionally, in this paper, we will consider special subclasses of formulas, namely
\begin{align*}
	\Gamma_{0, d} & =  \setdefinition{\ell_1\land\ldots\land \ell_c \mid \ell_1,\ldots, \ell_c \text{ are literals and }c\leq d },\\
	\Delta_{0, d} & =  \setdefinition{\ell_1\lor\ldots\lor \ell_c \mid \ell_1,\ldots, \ell_c \text{ are literals and }c\leq d },\\
	\Gamma_{t, d} & =  \left\{\,\bigwedge\limits_{i\in I} \alpha_i \,\middle|\, \alpha_i \in \Delta_{t-1,d}  \text{ for all } i \in I\, \right\}, 
	\Delta_{t, d}  =  \left\{\,\bigvee\limits_{i\in I} \alpha_i \,\middle|\, \alpha_i \in \Gamma_{t-1,d}  \text{ for all } i \in I\, \right\}.
\end{align*}
Finally, $\Gamma^{+}_{t,d}$ (resp. $\Gamma^{-}_{t,d}$) denote the class of all positive (negative) formulas in $\Gamma_{t,d}$.

\begin{example}
	Let $\phi = \bigwedge_{i\leq m}(\neg x_{i,1} \lor \cdots\lor \neg x_{i,n_i})$ for $1\leq n_i\leq d$ and $d,m\in\mathbb N$. 
	That is, $\phi$ is a conjunction of the clauses containing negative literals.
	Then $\phi \in \Gamma_{1,d}$, the so-called $d$-CNF. 
	Note also that $\phi$ is an ${\IS{d}{1}}$-formula using only negative clauses.
\end{example}

We will often reduce a problem instance to (and from) parameterised weighted satisfiability problem for propositional formulas.
This problem is defined below.
\paraproblemdef{$\wsat{t}{d}$}{A $\Gamma_{t,d}$-formula $\alpha$ over variables $V$ with $t, d \geq 1 $ and $k\in \mathbb{N}$}{$k$}{Is there a satisfying assignment for $\alpha$ of weight $k$}
Two similarly defined problems are $\wsatpos{t}$ and $\wsatneg{t}$ where an instance $\alpha$ comes from classes $  \Gamma^{+}_{t,1}$ (resp. $\Gamma^{-}_{t,1}$). 
The classes of the $\complClFont{W}$-hierarchy can be defined in terms of these problems as proved by Downey and Fellows~\cite{DBLP:series/txtcs/FlumG06}.
\begin{proposition}[\cite{DBLP:series/txtcs/FlumG06}]\label{theorem-wsat}
	For every $t \geq 1$, the following problems are \W t-complete under $\fptreduction$-reductions:
%	\begin{itemize}
		%\item 
		$\wsatpos{t} $ if t is even, 
		%\item 
		$\wsatneg{t} $ if t is odd, 
		%\item 
		$\wsat{t}{d}$ for every $t$ and $d\geq 1$.
%	\end{itemize}
\end{proposition}

\subparagraph*{Constraints and $\csl$-formulas.}
A \emph{logical relation} of arity $k$ is a relation $R \subseteq \{0,1\}^k$. 
A \emph{constraint} is a formula $R(x_1, \dots, x_k)$, where $R$ is a logical relation of arity $k$ and the $x_i$'s are (not necessarily distinct) variables. 
An assignment $\sigma$ to the $x_i$'s satisfies the constraint if $(\sigma(x_1), \dots, \sigma(x_k)) \in R$. 
A \emph{constraint language} $\csl$ is a finite set of logical relations. 
An \emph{$\csl$-formula $\varphi$} is a conjunction of constraints built upon logical relations only from $\csl$, and accordingly can be seen as a quantifier-free first-order formula. 
An assignment $\sigma$ is called a \emph{model} of $\varphi$ if $\sigma$ satisfies all constraints in $\varphi$ simultaneously. 
Whenever an $\csl$-formula or constraint is logically equivalent to a single clause or term, we treat it as such.

\begin{table}
	\centering
	\includegraphics[width=\linewidth]{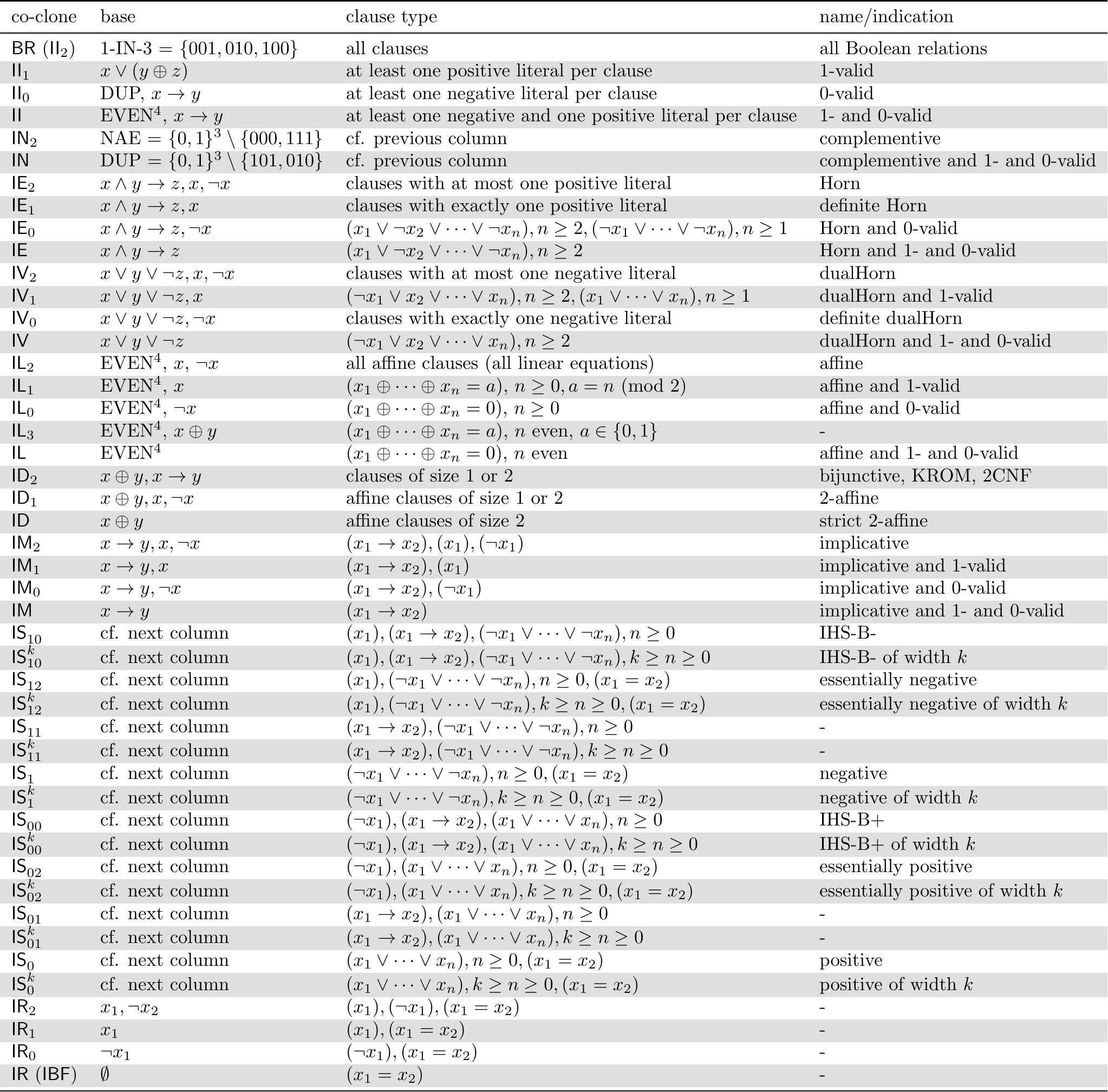}
	\caption{Overview of bases \cite{DBLP:journals/ipl/BohlerRSV05} and clause descriptions \cite{DBLP:journals/ai/NordhZ08} for co-clones, where
	% DUP = $\{0,1\}^3 \setminus \{101,101\}$ and
	EVEN$^4$ = $x_1 \oplus x_2 \oplus x_3 \oplus x_4 \oplus 1$.}\label{tab:bases}
\end{table}

%\subparagraph*{Definition.}
\begin{definition}
\begin{enumerate}
\item The set $\clos{\csl}$ is the smallest set of relations that contains $\csl$, the equality constraint, $=$, and which is closed under primitive positive first order definitions, that is, if $\phi$ is an $\csl \cup \{=\}$-formula and $R(x_1, \dots, x_n) \equiv \exists y_1 \dots \exists y_l \phi(x_1, \dots, x_n,y_1, \dots, y_l)$, then $R \in \clos{\csl}$. In other words, $\clos{\csl}$ is the set of relations that can be expressed as an $\csl \cup \{=\}$-formula with existentially quantified variables.
\item The set $\closneq{\csl}$ is the set of relations that can be expressed as an $\csl$-formula with existentially quantified variables (no equality relation is allowed).
\end{enumerate}
\end{definition}

The set $\clos{\csl}$ is called a \emph{relational clone} or \emph{co-clone} with \emph{base} $\csl$ \cite{DBLP:journals/ipl/BohlerRSV05}. Throughout the text, we refer to different types of Boolean relations and corresponding co-clones following Schaefer's terminology \cite{DBLP:conf/stoc/Schaefer78}.
For an overview of co-clones and bases, see Table~\ref{tab:bases}.
Note that $\closneq{\csl} \subseteq \clos{\csl}$ by definition. The other direction does not hold in general. However, if $(x = y) \in \closneq{\csl}$, then $\closneq{\csl} = \clos{\csl}$.

\subparagraph*{Abduction.}
An instance of the abduction problem for $\csl$-formulas is given by $\langle V, H, M, \KB \rangle $, where $V$ is the set of variables, $H$ is the set of hypotheses, $M$ is the set of manifestations, and $\KB$ is the knowledge base (or theory) built upon variables from $V$. 
A knowledge base $\KB$ is a set of $\csl$-formulas that we assimilate with the conjunction of all formulas it contains. We define the following abduction problems for $S$-formulas.
\paraproblemdef{$\abd{\csl}{k}$---the abductive reasoning problem for $\csl$-formulas parameterised by $k$}{$\langle V, H, M, \KB, k \rangle$, where $\KB$ is a set of $\csl$-formulas, $H, M$ are each set of propositions, and $V=\var(H)\cup\var(M)\cup\var(\KB)$}{$k$}{Is there a set $E\subseteq H$ such that $E\land\KB$ is satisfiable and $E\land\KB\models M$}
Similarly, the problem $\abd{\csl}{}$ is the classical pendant of $\abd{\csl}{k}$. 
Additionally, we consider size restrictions for a solution and define the following problems. 
\paraproblemdef{$\abdle{\csl}{k}$}{$\langle V, H, M, \KB, s, k \rangle$, where $\KB$ is a set of $\csl$-formulas, $H, M$ are each set of propositions, and $V=\var(H)\cup\var(M)\cup\var(\KB)$, and $s\in\mathbb N$}{$k$}{Is there a set $E\subseteq H$ with $|E|\leq s$ such that $E\land\KB$ is satisfiable and $E\land\KB\models M$}
Analogously, $\abdeq{\csl}{k}$ requires the size of $E$ to be exactly $s$ and $\abdeq{\csl}{}, \abdle{\csl}{}$ are the classical counterparts.
Notice that, for instance, in cases where the parameter is the size of solutions, then $s=k$.
\begin{example}
Sitting in a train you realise that it is still not moving even though the clock suggests it should be. 
You start reasoning about it. 
Either some door is open, the train has delayed, or that engine has failed. 
This form of reasoning is called abductive reasoning.
Having some additional information that the operator of train usually announces in case the train is delayed or engine has failed, you deduce that some door must be opened and that train will start moving soon when all the doors are closed.
Formally, one is interested in an explanation for the observed event (manifestation) $\setdefinition{ \texttt{stop}}$. 
The knowledge base includes following statements:
\begin{multicols}{2}
\begin{itemize}
	\item $\neg \texttt{moving} \leftrightarrow \texttt{stop}$
	\item $\neg\texttt{announcement}$,	
	\item $\texttt{moving} \rightarrow \texttt{time}$,
	\item $\texttt{engineFailed}\rightarrow\texttt{announcement}$,
	\item $\texttt{trainDelayed} \rightarrow \texttt{newTime}$, 
	\item $(\texttt{engineFailed} \lor \texttt{trainDelayed}  \lor \texttt{doorOpen} )\rightarrow \texttt{stop}$. 
\end{itemize}	
\end{multicols}

Then the set of hypotheses $\setdefinition{ \texttt{time}, \texttt{doorOpen}, \texttt{announcement}}$ has an explanation, namely, $\setdefinition{\texttt{doorOpen}}$.
On the other hand, $\setdefinition{\texttt{time}}$ does not explain the event $\setdefinition{ \texttt{stop}}$, whereas, $\setdefinition{\texttt{announcement}}$ is not consistent with the knowledge base.
Consequently, an explanation of size $1$ exists.
There also exists an explanation of size $2$ since $\setdefinition{\texttt{time}, \texttt{doorOpen}}$ is consistent with $\KB$ and explains $M$.
Note that having the set of hypotheses $\setdefinition{ \texttt{engineFailed}, \texttt{doorOpen}}$ facilitates only one explanation of size $1$, namely, $\setdefinition{ \texttt{doorOpen}}$, even though the hypotheses set has size $2$.
\end{example}

\subsection{Base Independence}\label{subsec:base-independence}
We present now a number of technical expressivity results (Lemma~\ref{lem:express_equality}).
They allow us in the sequel to prove a crucial property for the whole classification endeavour (Lemma~\ref{lem:independence}).
To prove the following lemma, we need to express equality by some other construction.
\begin{lemma}[$\star$]\label{lem:express_equality}
Let $\csl$ be a constraint language.
If $\csl$ is not essentially negative and not essentially positive, then $(x=y) \in \closneq{\csl}$, and $\clos{\csl} = \closneq{\csl}$.
\end{lemma}

The following property is crucial for presented results in the course of this paper.
It supplies generalised upper as well as lower bounds (independence of the base of a co-clone), as long as the constraint language is not essentially negative and not essentially positive.
The proof idea is to implement the previous lemma.
\begin{lemma}[$\star$]\label{lem:independence}
Let $\csl, \csl'$ be two constraint languages such that $\csl'$ is neither essentially positive nor essentially negative. 
Let $\abds{}{} \in \{\abd{}{}, \abdeq{}{}, \abdle{}{}\}$. If $\csl \subseteq \clos{\csl'}$, then
$\abds{\csl}{} \preduction \abds{\csl'}{}$. 
\end{lemma}

The last lemma in this section takes care of the essentially positive cases.
The proof idea is to remove the equality clauses maintaining the size counts and the satisfiability property.
\begin{lemma}[$\star$]\label{lem:base-ind-ess-positive}
Let $\csl, \csl'$ be two constraint languages such that $\csl'$ is essentially positive. 
Let $\abds{}{} \in \{\abd{}{}, \abdeq{}{}, \abdle{}{}\}$. 
If $\csl \subseteq \clos{\csl'}$, then $\abds{\csl}{} \preduction \abds{\csl'}{}$. 
\end{lemma}

\begin{remark}
	Notice that Lemmas~\ref{lem:independence} and~\ref{lem:base-ind-ess-positive} are stated with respect to the classical and unparameterised decision problem.
	However, these reductions can be generalised to $\fptreduction$-reductions whenever the parameters are bound as required by Def.~\ref{def:fpt-reduction}.
	That is, in our case, for any parameterisation $k\in\setdefinition{|H|,|E|,|M|}$ the reductions are valid.
	Even more, the values of the parameters stay the same as in the reduction the sizes of $H$, $E$, and $M$ remain unchanged.
\end{remark}

\begin{remark}
It is rather cumbersome to mention the base independence results in almost every single proof.
As a result, we omit this reference and show the results only for concrete bases, thereby, implicitly using the above lemmas. 
In cases where we deal with essentially negative constraint languages, we do not have a general base independence result, but direct constructions showing hardness in our cases for all bases (e.g., Lemma~\ref{lem:ABD-eq-E-IS12hk}).
\end{remark}

Let $\SAT$ and $\IMP$ denote the classical satisfiability and implication problems. 
Given a constraint language $\csl$ then an instance of $\SAT(\csl)$ is an $\csl$-formula $\varphi$ and the question is whether there exists a satisfying assignment for $\varphi$.
On the other hand, an instance of $\IMP(\csl)$ is $(\phi,\psi)$ such that $\phi, \psi$ are two $\csl$-formulas and the question is whether $\phi \models \psi$. 
We have the following observation regarding the classical $\SAT$ and $\IMP$ problems.

\begin{proposition}[\cite{DBLP:conf/stoc/Schaefer78,DBLP:conf/dagstuhl/SchnoorS08}]\label{prop:SAT-IMP-P}
Let $\csl$ be a constraint language such that $\clos\csl \subseteq \cocl$ where $\cocl \in \setdefinition{\ID_2,\IV_2, \IE_2, \IL_2}$. 
Then $\SAT(\csl)$ and $\IMP(\csl)$ are both in $\P$.
\end{proposition}

\section{Complexity results for abductive reasoning}
In this section, we start with general observations and reductions between the defined problems.
Then we prove some immediate (parameterised) complexity results.
We provide two results which help us to consider fewer cases to solve.
 
\begin{lemma}\label{lem:abd-logm-abdle}
	For every constraint language $S$ we have $\abd{\csl}{}\preduction\abdle{\csl}{}$.
\end{lemma}
\begin{proof}
	Clearly, $\abdinstance \in\abd{\csl}{}\Leftrightarrow \pabdformat{V,H,M,KB}{|H|}$.
	That is, there is an explanation for an abduction instance if and only if there is one with size at most that of the hypotheses set.
\end{proof}
\begin{lemma}\label{monotone}
	$\abdle{\csl}{}= \abdeq{\csl}{}$ for any $\csl$ such that $\IBF \subseteq \clos{\csl}  \subseteq \IV_2$. 
\end{lemma}
\begin{proof}
``$\subseteq$'':
	Every positive instance $\abdformat{V,H,M,\KB,s}\in\abdle{\csl}{}$ has a solution $E$ of size exactly $s$.
	We show that a solution of size $< s$ can be always extended to size $s$.
	Given a solution of size $\leq s$ then a solution of size $=s$ can be constructed from it (in even polynomial time) w.r.t.\ $|H|$ by adding one element $h$ at a time from $H$ to $E$ and checking that $\neg h \not \in \KB$.

``$\supseteq$'':	
	Every solution of size exactly $s$ is a solution of size $\le s$.
\end{proof}
%Lemma~\ref{monotone} allows us to consider only $\abdle{\csl}{}$ problem for the case $\IM \subseteq \clos{\csl}  \subseteq \IV_2$.
\subsubsection*{Intractable cases}
It turns out that for $0$-valid, $1$-valid and complementive languages, all three problems remain hard under any parametrisation except the case $|V|$.
\begin{lemma}\label{lem:ABD-H-hardness}
	The problems $\abd{\csl}{k}$, $\abdle{\csl}{k}$, $\abdeq{\csl}{k}$ are \begin{enumerate}
		\item $\para\co\NP$-hard if $\IN \subseteq \clos{\csl}\subseteq\II_1$ and any parameterisation $k\in \setdefinition{|H|,|E|, |M|}$,
		\item $\para\DP$-hard if $\cocl \subseteq \clos{\csl}\subseteq\BR$ and $\cocl\in\{\IN_2,\II_0\}$ and $k\in\{|H|,|E|\}$.
		\item $\para\SigmaP$-hard if $\cocl \subseteq \clos{\csl}\subseteq\BR$ and $k=|M|$ for $\cocl\in\{\IN_2,\II_0\}$.
	\end{enumerate} 
\end{lemma}
\begin{proof}
	(1.) We prove the case for $\IN$ regarding all three parameters simultaneously.
	Notice that $\IMP(\II_1)$ is $\co\NP$-hard \cite[Thm.~34]{DBLP:journals/ai/NordhZ08} even if the right side contains only a single variable.
We describe in the following a modified proof from \cite[Prop.~48]{DBLP:journals/ai/NordhZ08}.	
Since $\clos{\IN \cup \setdefinition{T}} = \II_1$ (define $T(x) \equiv x$) we have that $\IMP(\IN \cup \setdefinition{T})$ is $\co\NP\hard$, even if the right side contains only a single variable.
	We reduce $\IMP(\IN \cup \setdefinition{T})$ to our abduction problems with $|H|=1$, $|M|=1$, and $|E|=1$.
	Let $(\KB_T, q)$  be an instance of $\IMP(\IN \cup \setdefinition{T})$,
	where $\KB_T = \KB \land \bigwedge_{x\in V_T} T(x)$ with $\KB$ being an $\IN$-formula.
	We map $(\KB_T, q)$ to $\pabdformat{V, \{h\}, \{q\}, \KB'}{}{}$, where $V = \var(\KB) \cup \{h\}$, $h$ is a fresh variable, and $\KB'$ is obtained from $\KB$ by replacing any variable from $V_T$ by $h$. 
	Note that $\KB_T \equiv \KB' \land h$.
	Since $\KB$ and $\KB'$ are 1-valid, clearly, $\KB' \land h$ is always satisfiable and there exists an explanation iff $\KB' \land h \models q$, iff $\KB_T \models q$.
	Furthermore, observe that $\KB_T\models q$ if and only if
	$\pabdformat{V, \{h\}, \{q\}, \KB'}{|H|}{} \in\abd{\IN}{|H|}$
	 if and only if
	$\pabdformat{V, \{h\}, \{q\}, \KB',1}{|H|} \in\abdle{\IN}{|H|}$
	 if and only if
	$\pabdformat{V, \{h\}, \{q\}, \KB',1}{|H|} \in\abdeq{\IN}{|H|}$. 
	The latter is true also when replacing $|H|$ by $|E|$ or $|M|$.
This proves the claimed $\para\co\NP$-hardnesses. 
%%%%%%%% old and wrong
	%(1.) We prove the case for $\IN$ regarding all three parameters simultaneously.
	%Notice that $\IMP(\II_1)$ is $\co\NP$-hard \cite[Thm.~34]{DBLP:journals/ai/NordhZ08} even if the right side contains only a single variable.
	%Furthermore, this result is true even if we replace $\II_1$ by $\IN$.
	%That is, $\clos{\IN \cup \setdefinition{T}} = \II_1$ and consequently, $\IMP(\IN)$ is $\co\NP\hard$ (define $T(x) \equiv x$). 
	%We reduce $\IMP(\IN)$ to our abduction problems with $H=\emptyset$, $M=\{x\}$ and $|E|=0$.
	%Then, $\KB\models x$ if and only if $\pabdformat{V, \emptyset, \{x\}, \KB}{|H|}{} \in\abd{\II_1}{|H|}$ if and only if $\pabdformat{V, \emptyset, \{x\},\KB,0}{|H|}{} \in\abdle{\II_1}{|H|}$ if and only if $\pabdformat{V, \emptyset, \{x\},\KB,0}{|H|}{}\in\abdeq{\IN}{|H|}$.
%This proves the claimed $\para\co\NP$-hardnesses. 

%	Similar arguments work for $\para\SigmaP$-hardness since $\IN_2 \cup \setdefinition{T} = \BR$ and $\abd{\BR}{}$ is $\SigmaP\hard$.
%	\todo[inline]{I am not sure about $\SigmaP$ anymore. At least the above trick (reduction from IMP-problem) does not work for H. However, for M everything is fine.}
(2.) From Fellows~et~al.~\cite[Prop.~4]{DBLP:conf/aaai/FellowsPRR12} we know that all three problems for $\BR$ are $\DP$-complete for $|H|=0$ even if $|M|=1$. 
We argue that the hardness can be extended to $\IN_2$.
Note that $\clos{\IN_2\cup\{F\}}=\BR$ where $F(x) \equiv \neg x$.
Creignou \& Zanuttini \cite{DBLP:journals/siamcomp/CreignouZ06} prove that $\abd{\csl \cup \setdefinition{F}}{}\preduction \abd{\csl \cup \setdefinition{\texttt{SymOR}_{2,1}}}{} $ where $\texttt{SymOR}_{2,1}(x,y,z)= ((x\rightarrow y) \land T(z)) \lor ((y\rightarrow x) \land F(z))$. 
Moreover, they also prove that $\texttt{SymOR}_{2,1} \in \clos{\csl}$ such that $\IN_2\subseteq \clos{\csl}$ \cite[Lem.~21,27]{DBLP:journals/siamcomp/CreignouZ06}.
Finally, having $|M|=1$ allows us to use their proof and, as a consequence, $ \abd{\BR}{}  \preduction \abd{\csl}{} $ such that $\IN_2 \subseteq \clos\csl$. 
This gives the desired lower bound for $\IN_2$.
Regarding $\II_0$, the proof follows by similar arguments using the observations that $\clos{\II_0 \cup \setdefinition{T}} = \BR$ and $\texttt{OR}_{2,1} \in \clos\csl $ such that $\II_0\subseteq \clos\csl$ where $\texttt{OR}_{2,1}(x,y) = x\rightarrow y$ \cite[Lem.~19 and 27]{DBLP:journals/siamcomp/CreignouZ06} .

(3.) Nordh \& Zanuttini \cite[Prop.~46/47]{DBLP:journals/ai/NordhZ08} prove $\SigmaP$-hardness for both $\IN_2$ as well as $\II_0$  with positive literal manifestations.
This implies that the $1$-slice of each of $\abd{\IN_2}{|M|}$ and $\abd{\II_0}{|M|}$ is $\SigmaP$-hard, which gives the desired result.
For $\abdle{S}{|M|}$ and $\abdeq{S}{|M|}$, the results follow from Lemma~\ref{lem:abd-logm-abdle}.
\end{proof}

\subsubsection*{Fixed-parameter tractable cases}
The following corollary is immediate because the classical questions corresponding to these cases are in $\P$ due to Nordh and Zanuttini~\cite{DBLP:journals/ai/NordhZ08}.
\begin{corollary}\label{lem:ABD-any-IS12}
	The problem $\abd{\csl}{k}$ is $\FPT$ for any parameterisation $k$ and $\clos\csl\subseteq\cocl$ with $\cocl \in \setdefinition{\IV_2, \ID_1, \IE_1, \IS{}{12}} $.
\end{corollary}
\noindent The next result is already due to Fellows~et~al.~\cite[Prop.~13]{DBLP:conf/aaai/FellowsPRR12}.
\begin{corollary}\label{cor:variables-all}
The problems $\abd{\csl}{|V|}$, $\abdle{\csl}{|V|}$, $\abdeq{\csl}{|V|}$ are all $\FPT$ for all Boolean constraint languages $\csl$.
\end{corollary}
%\todo[inline]{Following might not be needed} 
%Note that the Galois correspondence does not hold when considering the parameter $|V|$. 
%This is due to the reason that for $\clos{\csl'}\subseteq \clos{\csl}$, the positive primitive definition of a $\KB' \in \clos{\csl'}$ by a $\KB \in \clos{\csl}$ may increases the number of variables, thereby not bounding the new parameter value by a function of the old parameter. 
%However, the brute force argument underlying the corollary~\ref{cor:variables-all} still holds for any constraint language $\csl$, which allows us to prove the claim.    

Now, we prove $\P$-membership for some cases of the classical problems and start with the essentially positive cases.
The proof idea is to start with unit propagation. 
The positive clauses do not explain anything and one just only checks whether the elements of $M$ appear either in $\KB$ or $H$.
Then, we need to adjust the size accordingly.
\begin{lemma}[$\star$]\label{lem:ABD-IS02}
	The classical problems \abdeq{\csl}{} and \abdle{\csl}{} are in \P for $\clos\csl\subseteq\IS{}{02}$.
\end{lemma}

The following lemma proves that essentially negative languages for $\abdle{}{}$ also remain tractable.

\begin{lemma}\label{lem:ABD-leq-E-IS12}
	The classical problem $\abdle{\csl}{}$ is in \P if $\clos\csl\subseteq\IS{}{12}$.
\end{lemma}
\begin{proof}
First, we prove the result with respect to $\closneq{S}\subseteq\IS{}{12}$.
Let $P$ denote the set of positive unit clauses from $\KB$ and denote $E_{MP} = M \setminus P$. 
Now, we have the following two observations.
\begin{description}
	\item[Observation 1] There exists an explanation iff $E_{MP} \subseteq H$ and $M$ is consistent with $\KB$.
That is, what is not yet explained by $P$ must be explainable directly by $H$ because  negative clauses can not contribute to explaining anything, they can only contribute to `rule out' certain subsets of $H$ as possible explanations.

	\item[Observation 2] If there exists an explanation, then any explanation contains $E_{MP}$. 

\end{description}
As a result, $E_{MP}$ represents a cardinality-minimal and a subset-minimal explanation.
We conclude that there exists an explanation $E$ with $|E| \leq s$ iff $E_{MP}$ constitutes an explanation and $|E_{MP}| \leq s$.
Now, we proceed with base independence for this case.
\begin{claim}
$\abdle{\csl \cup \{=\}}{}\preduction\abdle{\csl}{}$ for $\clos\csl\subseteq\IS{}{12}$.
\end{claim}
\begin{claimproof}
The reduction gets rid of the equality clauses by removing them and deleting the duplicating occurrences of variables.
This decreases only the size of $H$ and might also the size of an explanation $E$.
Notice that $x=y\in\KB$ does not enforce both $x$ and $y$ into $E$.
\end{claimproof}
This completes the proof to lemma.
\end{proof}
Finally, the $2$-affine cases are also tractable as we prove in the following lemma.
The idea is, similarly to Creignou et~al.~\cite[Prop.~1]{DBLP:conf/sat/CreignouOS11}, to change the representation of the knowledge base. 
\begin{lemma}[$\star$]\label{lem:ABD-E-D1}
	The classical problems $\abdeq{\csl}{}$ and $\abdle{\csl}{}$ are in \P if $\clos\csl\subseteq\ID_1$. %[2-affine]
\end{lemma}
\subsection{Parameter `number of hypotheses' |H|}
For this parameter, it turns out that the only intractable cases are those pointed out in Lemma~\ref{lem:ABD-H-hardness}.  
\begin{theorem}\label{theorem:H}
	$\abd{\csl}{|H|}$, $\abdle{\csl}{|H|}$ and $\abdeq{\csl}{|H|}$ are
\begin{multicols}{2}
	\begin{enumerate}\nolinenumbers
		\item $\para\DP$-hard if $\cocl \subseteq \clos{\csl}\subseteq\BR$\\ and $\cocl\in\{\IN_2,\II_0\}$,
		\item $\para\co\NP$-hard if $\IN \subseteq \clos{\csl}\subseteq\BR$,\label{thm:H}
		\item $\FPT$ if $\clos\csl\subseteq\cocl \in \setdefinition{\IE_2, \IV_2, \ID_2,\IL_2}$.
	\end{enumerate}\linenumbers
\end{multicols}
\end{theorem}
\begin{proof}
	\begin{description}
		\item[1.+2.] We proved these cases in Lemma~\ref{lem:ABD-H-hardness}.
		\item[3.] Recall that $\SAT(\csl) \text{ and } \IMP(\csl)$ are both in $\P$ for every $\csl$ in the question (Prop.~\ref{prop:SAT-IMP-P}).
	By $|H|\ge|E|$, we have that $\binom{|H|}{|E|}=|H|^{|E|}\in O(k^k)$, where $k=|H|$.
	Consequently, we brute-force the candidates for $E$ and verify them in polynomial time.
	This yields $\FPT$ membership.\qedhere
\end{description}	
\end{proof}

%\begin{proof}
%	There are fpt-many explanations, as $|V|\ge|H|\ge|E|$: so ${H\choose E}\in|V|^{|E|}\in O(k^k)$ where $k=|V|$.
%	As $|V|$ is finite, we can brute-force the satisfiability- and implication-question for each candidate in fpt-time.
%	Summarising, this yields a (rough) upper bound runtime of
%	$$
%	2^{|V|}\cdot2^{|V|}\cdot2^{|V|}\cdot n^{O(1)},
%	$$
%	where $n$ is the input length.
%	In a sense, this is a nested brute-force approach.
%	This shows membership in $\FPT$.
%	As we only have fpt-many explanations, enumerating them is possible in $\FPT\enum$.
%%Following versio copied from Pfandler 12
%%	There are at most $2^{|H|} \leq 2^{|V|} $ possible solution candidates.
%%	For each of them we need to test consistency and entailment. 
%%	This can be done in time $O(2^{|H|}(n+2^{|V|} n))$.
%\end{proof}

% OLD XP-MEMBERSHIP PROOF
%\begin{lemma}\label{lem:ABD-E-IL3}
%	$\abdle{\IL_3}{|E|}$, $\abdeq{\IL_3}{|E|}$ are in $\XP$.
%\end{lemma}
%\begin{proof}
%	We just know an $\XP$-upper bound for the number of explanation candidates: ${H\choose E}\in O(n^k)$, where $k=|E|$.
%	Then, we just need solve implication which is easy (see Lemma~\ref{lem:ABD-H-IL3}). 
%	So this yields an $\XP$ upper bound.
%\end{proof}

\subsection{Parameter `number of explanations' |E|}
In this subsection, we consider the solution size as a parameter.
Notice that, because of the parameter $|E|$, the problem $\problemFont{ABD}$ is not meaningful anymore.
As a result, we only consider the size limited variants $\problemFont{ABD}_=$ and $\problemFont{ABD}_\leq$.
The following theorem provides a classification into six different complexity degrees.

\begin{theorem}\label{theorem-E}
	The problems $\abdle{\csl}{|E|}$ and $\abdeq{\csl}{|E|}$ are
\begin{multicols}{2}\nolinenumbers
	\begin{enumerate}
		\item $\para\DP$-hard if $\cocl \subseteq \clos{\csl}\subseteq\BR$\\ and $\cocl\in\{\IN_2,\II_0\}$
		\item $\para\co\NP$-hard if $\IN \subseteq \clos{\csl}\subseteq\II_1$,
		\item $\W{\P}$-complete if $\IE \subseteq \clos{\csl}\subseteq\IE_2$, 
%		\item in $\W{\P}$ if $\IL \subseteq \clos{\csl}\subseteq\IL_2$, 
		\item $\W{2}$-complete if $\IM \subseteq \clos{\csl}\subseteq \cocl$\\ for 
			$\cocl \in \{\ID_2, \IV_2\}$ and $\W{2}$-hard\\ if $\IM \subseteq \clos\csl \subseteq \IS{}{10}$,
		\item $\FPT$ if $\clos\csl \subseteq \ID_1$ or $\clos\csl \subseteq \IS{}{02}$,
	\end{enumerate}
\linenumbers
\end{multicols}
\nolinenumbers \noindent  Moreover, if $\IS{2}{1} \subseteq \clos{\csl}\subseteq\IS{}{12} $, then $\abdle{\csl}{|E|}\in\FPT$ and $\abdeq{\csl}{|E|}$ is $\W1$-complete.
\end{theorem}
\begin{proof}[Proof Ideas.]
	\begin{description}
		\item[1.+2.] This is a corollary to Theorem~\ref{theorem:H}.
		\item[3.] Upper bound for $\IE_2$ follows from the fact that $\SAT(\IE_2)$ and $\IMP(\IE_2)$ are in $\P$ (cf.\ Prop.~\ref{prop:SAT-IMP-P}).
		Guessing $E$ takes $k \cdot \log n$ non-deterministic steps and verification can be done in polynomial time.
		For the lower bound, we argue that the proof in \cite[Thm.~8]{DBLP:conf/aaai/FellowsPRR12} can be extended. 
		Details are presented in Lemma~\ref{lem:ABD-E-IE}.
%		\item[] Follows by the same arguments as above. See Lemma~\ref{lem:ABD-E-IL3}. 
		\item[4.] 
		Note that the difficult part of the abduction problem for $\clos\csl$ such that $\IM \subseteq \clos\csl$ is the case when a solution of size larger than $k$ is found. 
		This solution must be reduced to one of size $\leq k$ (resp. $=k$). 
		For $\W2$-membership of $\abdeq{\IM}{|E|}$, we reduce our problem to $\wsat{2}{1}$ which is $\W2$-complete. 
		For hardness, we reduce from $\wsatpos{2}$ which is again \W2-complete.
		Details of the completeness proof for $\abdeq{\IM}{E}$ can be found in Lemma~\ref{lem:ABD-E-IM-W2}. 
		The $\W2$-membership for $\IV_2$ uses a little modification of the same reduction. 
		This is proved in Lemma~\ref{lem:ABD-E-IV2-inW2}. 
		For these two cases, $\abdle{\csl}{|E|}$ follow from the monotone argument from Lemma~\ref{monotone}.
		For $\ID_2$, the result follows from \cite[Thm. 21]{DBLP:conf/aaai/FellowsPRR12}.
		Finally, the hardness for $\IS{}{10}$ is a consequence of the $\W2$-hardness for \IM.
		However, Lemma~\ref{lem:ABD-E-IS10-inW2} strengthens this results to $\W2$-completeness by showing membership in $\W2$ for $\abdeq{}{}$.
		Regarding $\abdle{\IS{}{10}}{|E|}$, we also believe in $\W{2}$-completeness but have not proved it yet.
		\item[5.] This follows from the fact that the classical problems are in $\P$ (Lemmas~\ref{lem:ABD-IS02} and \ref{lem:ABD-E-D1}).
	\end{description}\bigskip
	
Finally, $\FPT$ membership for $\abdle{\IS{}{12}}{|E|}$ follows from Lemma~\ref{lem:ABD-leq-E-IS12}. 
Note that this is the only case with $|E|$ when the two problems $\abdle{}{}$ and $\abdeq{}{}$ have different complexity. 
We prove $\W{1}$-hardness for the languages $\csl$, such that $\neg x\lor \neg y \in \closneq \csl$. 
The membership for $\abdeq{S}{|E|}$ with $\clos\csl\subseteq\IS{}{12}$, this means also arbitrary bases, then follows as a corollary (for details see Lemmas~\ref{lem:ABD-eq-E-IS1h2-hardness}~and~\ref{lem:ABD-eq-E-IS12hk}). %\todo[inline]{For completeness, (W1-upper bound) the trick mentioned before Lemma32 or may not work}
%
% $\abdeq{\IS{}{12}}{|E|}$ by reducing from $\IndSet$ and to $\wsat{1}{\ell}$ 
\end{proof}

%%%%%%%%%%%%%%%%%%%%%%%%%%%%%%%%%%%%%%%%%%%%%%%%%%%%%%%%%%%%%%%
%\begin{figure}
%	\centering
%	\includegraphics[width=\linewidth]{posts-lattice-equal-E}
%	\caption{Complexity landscape of $\abdeq{\Gamma}{|E|}$.}
%\end{figure}
%
%\begin{figure}
%	\centering
%	\includegraphics[width=\linewidth]{posts-lattice-leq-E}
%	\caption{Complexity landscape of $\abdle{\Gamma}{|E|}$.}
%\end{figure}
%
%
%%Added the 2nd lattice for case \abd{\Gamma}{|M|}
%\begin{figure}
%	\centering
%	\includegraphics[width=\linewidth]{posts-lattice-general-M}
%	\caption{Complexity landscape of $\abd{\Gamma}{|M|}$.}
%\end{figure}
%
%\begin{figure}
%	\centering
%	\includegraphics[width=\linewidth]{posts-lattice-equal-M}
%	\caption{Complexity landscape of $\abd{\Gamma}{|M|}$.}
%\end{figure}
%
%
%\begin{figure}
%	\centering
%	\includegraphics[width=\linewidth]{posts-lattice-leq-M}
%	\caption{Complexity landscape of $\abd{\Gamma}{|M|}$.}
%\end{figure}
%
%\begin{lemma}\label{lem:ABD-E-IS211}
%	$\abdeq{\IS{2}{11}}{|M|}$ is $\W1$-hard. \todo{Check if this is repeating}
%\end{lemma}
%\begin{proof}
%	We encode the independent set problem which is $\W1$-complete \cite{DBLP:series/txtcs/FlumG06}.
%	$$
%	(V,\tilde E,k)\in\IndSet \Leftrightarrow (KB=\bigwedge_{\{u,v\}\in\tilde E}\NAND(u,v)\land\AND(a,b),H=V,M=\{a,b\},k).
%	$$
%\end{proof}
%
%\begin{corollary}\label{cor:ABD-E-M-IS211}
%	$\abdle{\IS{2}{11}}{|E|+|M|}$, $\abdeq{\IS{2}{11}}{|E|+|M|}$ are in $\W\P$ and $\abdeq{\IS{2}{11}}{|E|+|M|}$ is $\W1$-hard.
%\end{corollary}

\subsection{Parameter `number of manifestations' |M|}
The complexity landscape regarding the parameter $|M|$ is more diverse. 
The classification differs for each of the investigated problem variants.
Consequently, we treat each case separately and start with the general abduction problem which provides a hexachotomy.
\begin{theorem}\label{theorem:g-M}
	The problem $\abd{\csl}{|M|}$ is 
	\begin{multicols}{2}\nolinenumbers
		\begin{enumerate} 
		\item $\para\SigmaP$-hard if $\cocl \subseteq \clos{\csl}\subseteq\BR$\\ and $\cocl\in\{\IN_2,\II_0\}$,
		\item $\para\co\NP$-hard if $\IN \subseteq \clos{\csl}\subseteq\II_1$,
		\item $\para\NP$-complete if $ \clos{\csl}=\IE_2 $,
		\item $\W{1}$-complete if $\IS{2}{11} \subseteq \clos{\csl}\subseteq\ID_2$,
		\item $\W{1}$-hard if $\IS{3}{11}\subseteq \clos{\csl}$,
		\item $\FPT$ if $\clos\csl\subseteq\cocl \in \{\ID_1, \IS{}{12},\IE_1,\IV_2\}$.
	\end{enumerate}\linenumbers
\end{multicols}
\end{theorem}
\begin{proof}
	\begin{description}
		\item[1.+2.] We proved this in Lemma~\ref{lem:ABD-H-hardness} using the fact that 1-slice of each problem is hard for respective classes.
		\item[3.] Membership is easy to see since the classical problem is \NP-complete. 
		For hardness, notice that the $1$-slice of the problem is \NP-complete \cite{DBLP:journals/jacm/EiterG95}.
		\item[4.+5.] The first result follows from Fellows~et~al.~\cite[Thm.~26]{DBLP:conf/aaai/FellowsPRR12}. 
		Notice that they prove this for $\ID_2$, but using the fact that the formulas (or clauses) in their reduction are $\IS{2}{11}$-formulas, we derive the hardness for $\IS{2}{11}$.
		The second statement is then a consequence.
		\item[6.] Follows from classical problems being in \P (Corollary~\ref{lem:ABD-any-IS12}).\qedhere
	\end{description}
\end{proof}

For $\abdle{}{}$, definite Horn cases surprisingly behave different and are much harder than for the general case. 
\begin{theorem}\label{theorem:le-M}
	The problem $\abdle{\csl}{|M|}$ is 
	\begin{multicols}{2}\nolinenumbers
		\begin{enumerate}
		\item $\para\SigmaP$-hard if $\cocl \subseteq \clos{\csl}\subseteq\BR$\\ and $\cocl\in\{\IN_2,\II_0\}$,
		\item $\para\co\NP$-hard if $\IN \subseteq \clos{\csl}\subseteq\II_1$,
		\item $\para\NP$-complete if $\IE \subseteq \clos{\csl}\subseteq \IE_2$, 
		\item $\W{1}$-complete if $\IS{2}{11} \subseteq \clos{\csl}\subseteq\ID_2$,
		\item $\W{1}$-hard if $\IS{3}{11}\subseteq \clos{\csl}$,
		\item $\FPT$ if $\clos\csl\subseteq\cocl\in \{\ID_1, \IS{}{12}, \IV_2\}$.

	\end{enumerate}\linenumbers
\end{multicols}
\end{theorem}
\begin{proof}[Proof Ideas.]
	\begin{description}
		\item[1.+2.] Follows from Theorem~\ref{theorem:g-M} in conjunction with Lemma~\ref{lem:abd-logm-abdle}.
		\item[3.] We reduce \vcover\ to our problem similar to the approach of Fellows~et~al.~\cite[Thm.~5]{DBLP:conf/aaai/FellowsPRR12}.
		The problem can be translated into an abduction instance with $\IE$ knowledge base, consequently giving the desired hardness result.
		\item[4.+5.] The first result follows from \cite[Thm.~25]{DBLP:conf/aaai/FellowsPRR12}. 
		Notice that they prove this for $\ID_2$, but using the fact that the formulas (or clauses) in their reduction are $\IS{2}{11}$-formulas, we derive a hardness result for $\IS{2}{11}$. 
		The second statement is then a consequence.
		\item[6.] We prove this for $\IM$ by reducing our problem to the $\maxsat$ problem which asks, given $m$ clauses, is it possible to set at most $s$ variables to true so that at least $k$ clauses are satisfied (details are presented in Lemma~\ref{lem:ABD-M-IM-FPT}).
		This problem when parametrised by $k$, the number of clauses to be satisfied, is $\FPT$.
		Moreover, this reduction can be extended to the languages in $\IV_2$.
		The problematic part is the presence of positive and unit negative clauses which need to be taken care of (for details, see Lemma~\ref{lem:ABD-M-IV2-FPT}).
		Accordingly, the result for $\IV_2$ follows.
		The remaining cases are due to Lemmas~\ref{lem:ABD-leq-E-IS12} and \ref{lem:ABD-E-D1}.\qedhere
	\end{description}
\end{proof}

Now, we end by stating results for $\problemFont{ABD}_=$.
Interestingly to observe, the majority of the intractable cases is already much harder with large parts being $\para\NP$-complete.
Even the case of the essentially negative co-clones which are \FPT for $\problemFont{ABD}_\leq$ yield $\para\NP$-completeness in this situation.
Merely the $2$-affine and dualHorn cases are $\FPT$.
\begin{theorem}\label{theorem:E-M}
	The problem $\abdeq{\csl}{|M|}$ is 
	\begin{multicols}{2}\nolinenumbers
		\begin{enumerate}
		\item $\para\SigmaP$-hard if $\cocl \subseteq \clos{\csl}\subseteq\BR$\\ and $\cocl\in\{\IN_2,\II_0\}$,
		\item $\para\co\NP$-hard if $\IN \subseteq \clos{\csl}\subseteq\II_1$,
		\item $\para\NP$-complete if $\IS{2}{1} \subseteq \clos{\csl}$ and $\clos\csl\subseteq\cocl \in \{\IE_2, \ID_2\}$,
		\item $\FPT$ if $\clos\csl\subseteq\cocl \in \{\ID_1, \IV_2\}$.
	\end{enumerate}\linenumbers
\end{multicols}
\end{theorem}
\begin{proof}[Proof Ideas.]
	\begin{description}
		\item[1.+2.] Follows from Theorem~\ref{theorem:g-M} in conjunction with Lemma~\ref{lem:abd-logm-abdle}.

		\item[3.] In Lemma~\ref{lem:ABD-eq-M-IS1h2-hardness}, we prove that the problem $\abdeq{\csl}{|M|}$ is $\para\NP$-hard as long as $\neg x\lor \neg y \in \closneq \csl$.
		 The case for $\abdeq{\IS{2}{1}}{|M|}$, so also arbitrary bases, then follows as a corollary. 
	The hardness for $\IE\subseteq \clos{\csl}$ follows from arguments used in the proof of Theorem~\ref{theorem:le-M} for the $\IE$ case.
	The upper bounds for $\IE_2$ and $\ID_2$ follow trivially since the classical problems are in \NP. 
		\item[4.] The proof for $\IV_2$ is due to the monotone argument of Lemma~\ref{monotone} and Theorem~\ref{theorem:le-M}.
		For $\ID_1$, we proved in Lemma~\ref{lem:ABD-E-D1} that the classical problem is in \P.\qedhere
	\end{description}
\end{proof}

\section{Conclusion}
In this paper, we presented a two-dimensional classification of three central abductive reasoning problems (unrestricted explanation size, $=$, and $\leq$).
In one dimension, we consider the different parameterisations $|H|,|M|,|V|,|E|$, and in the other dimension we consider all possible constraint languages defined by corresponding co-clones but the affine co-clones.
Often in the past, problems regarding the affine co-clones (resp., clones) resisted a complete classification \cite{bamuscscscvo11,crscthwo10,DBLP:conf/kr/CreignouST10,thomas09,DBLP:journals/corr/abs-0812-4848,MMSTWW09,hescsc08,DBLP:phd/de/Reith2002}.
Also the result of Durand and Hermann \cite{DBLP:conf/stacs/DurandH03} underlines how restive problems around affine functions are.
It is difficult to explain why exactly these cases are so problematic but the notion of the Fourier expansion \cite{DBLP:books/daglib/0033652} of Boolean functions gives a nice and fitting view on that.
Informally, the Fourier expansion of a Boolean function is a probability measure mimicking how likely a flip of a variable changes the function value.
For instance, disjunctions have a very low Fourier expansion value whereas the exclusive-or function has the maximum. 
Affine functions can though be seen as rather counterintuitive as every variable influences the function value dramatically.

For all three studied problems, we exhibit the same trichotomy for the parameter $|H|$ ($\IN$ is $\para\co\NP$-hard, $\IN_2$ is $\para\DP$-hard, and the remaining are $\FPT$).
The parameter $|V|$ always allows for \FPT algorithms independent of the co-clone.
Regarding $|E|$, only the two size restricted variants are meaningful.
For `$\leq$' we achieve a pentachotomy between $\FPT$, $\W2$-hard, $\W\P$-complete, $\para\co\NP$-, and $\para\DP$-hard.
Whereas, for `$=$', we achieve a hexachotomy additionally having $\W1$-hardness for the essentially negative cases.
These $\W1$-hard cases are also surprising in the sense that for `$\leq$' they are easy and $\FPT$.
Similarly, the same easy/hard-difference has been observed as well for $|M|$ as the studied parameter.
However, here, we distinguish between $\para\NP$-complete for `$=$' and $\FPT$ for `$\leq$'.
The complete picture for `$=$' and $|M|$ is a tetrachotomy ranging through $\FPT$, $\para\NP$-complete, $\para\co\NP$-hard, and $\para\SigmaP$-complete.
With respect to `$\leq$' and the unrestrictied cases, we also have some $\W1$-hard cases which lack a precise classification.

Additionally, we already started a bit to study the parameterised enumeration complexity~\cite{DBLP:journals/mst/CreignouMMSV17} of these problems yielding $\FPT\enum$ algorithms for $|V|$ and $\BR$ as well as for $|H|$ and $\IE_2,\IV_2,\ID_2$, and $\IL_2$.
Furthermore, $\IL_1$ even allows \FPT algorithms for any parameterisation (so it extends Corollary~\ref{lem:ABD-any-IS12} in that way).

Notice that in this paper, we did not require $H\cap M$ to be empty. 
However, one can require this (as, for instance, Fellows~et~al.~\cite{DBLP:conf/aaai/FellowsPRR12} did). 
All our proofs (e.g., Lemma~\ref{lem:ABD-IS02}) can easily be adapted in that direction.
Furthermore, we believe that the $\para\DP$-hardness for $|H|$ and $\IN_2$ should be extendable to $\para\SigmaP$-hardness but do not have a full proof yet.

Furthermore, we want to attack the affine co-clones as well as present matching upper and lower bounds for all cases.
Also, parameterised enumeration complexity is the next object of our investigations.

\appendix
\bibliography{main}

\section{Omitted proof details}
\subsection{Base Independence}
\begin{restatelemma}[lem:express_equality]
\begin{lemma}
	Let $\csl$ be a constraint language.
If $\csl$ is not essentially negative and not essentially positive, then $(x=y) \in \closneq{\csl}$, and $\clos{\csl} = \closneq{\csl}$.
\end{lemma}
\end{restatelemma}

\begin{proof}
Let us start with a needed proposition in the proof.

\begin{proposition}\label{prop:express_equality_literature}(\cite{DBLP:journals/tocl/CreignouE014})
Let $\csl$ be a constraint language. The following is true:
\begin{enumerate}

% ToCL L4.6 1st item
\item If $\csl$ is not 1-valid, not 0-valid, and complementive, then $(x\neq y) \in \closneq{\csl}$ \cite[Lem.~4.6.1]{DBLP:journals/tocl/CreignouE014}.

% ToCL L4.6 3rd item
\item If $\csl$ is not 1-valid, not 0-valid, and not complementive, then $(x\land \neg y) \in \closneq{\csl}$ \cite[Lem.~4.6.3]{DBLP:journals/tocl/CreignouE014}.

% ToCl L4.7
\item If $\csl$ is 1-valid, 0-valid, and not trivial, then $(x=y) \in \closneq{\csl}$ \cite[Lem.~4.7]{DBLP:journals/tocl/CreignouE014}.

% ToCl L4.8 1st item
\item If $\csl$ is 1-valid, not 0-valid, and not essentially positive, then $(x=y) \in \closneq{\csl}$ \cite[Lem.~4.8.1]{DBLP:journals/tocl/CreignouE014}.

% ToCl L4.8 2st item
\item If $\csl$ is not 1-valid, 0-valid, and not essentially negative, then $(x=y) \in \closneq{\csl}$ \cite[Lem.~4.8.2]{DBLP:journals/tocl/CreignouE014}.
\end{enumerate}
\end{proposition}

For constraint languages (relations) that are Horn ($\IE_2$), dualHorn ($\IV_2$), essentially negative ($\IS{}{12}$) or essentially positive ($\IS{}{02}$) we use the following characterisations by polymorphisms (see, e.g., the work of Creignou and Vollmer~\cite{DBLP:conf/dagstuhl/CreignouV08}). 
The binary operations of conjunction, disjunction and negation are applied coordinate-wise.
\begin{enumerate}
\item $R$ is Horn if and only if $m_1,m_2 \in R$ implies $m_1 \land m_2 \in R$.
\item $R$ is dualHorn if and only if $m_1,m_2 \in R$ implies $m_1 \lor m_2 \in R$.
\item $R$ is essentially negative if and only if $m_1,m_2,m_3 \in R$ implies $m_1 \land (m_2 \lor \neg m_3) \in R$.
\item $R$ is essentially positive if and only if $m_1,m_2,m_3 \in R$ implies $m_1 \lor (m_2 \land \neg m_3) \in R$.
\end{enumerate}

Now we start with the proof of the lemma.
We make a case distinction according to whether $\csl$ is 1- and/or 0-valid.
\begin{description}
	\item[1-valid and 0-valid.] This case follows immediately from Prop.~\ref{prop:express_equality_literature}, 3rd item.
	\item[1-valid and not 0-valid.] This case follows immediately from Prop.~\ref{prop:express_equality_literature}, 4th item.
	\item[not 1-valid and 0-valid.] This case follows immediately from Prop.~\ref{prop:express_equality_literature}, 5th item.
	\item[not 0-valid and not 1-valid.] We make another case distinction according to whether $\csl$ is Horn and/or dualHorn.
	\begin{description}
\item[not Horn and not dualHorn.] It suffices here to show that inequality $(x \neq y)$ can be expressed, since $(x=y) \equiv \exists z(x\neq z) \land (z\neq y)$. If $\csl$ is complementive, we obtain by Prop.~\ref{prop:express_equality_literature}, 1st item, that $(x\neq y) \in \closneq{\csl}$. Therefore suppose now that $\csl$ is not complementive.

Let $R$ be a relation that is not Horn. Then there are $m_1,m_2 \in R$ such that $m_1 \land m_2 \notin R$. For $i,j \in \{0,1\}$, set $V_{i,j} = \{x \mid x \in V,\ m_1(x) = i,\ m_2(x) = j\}$.
Observe that the sets $V_{0,1}$ and $V_{1,0}$ are nonempty (otherwise $m_1 = m_1 \land m_2$ or $m_2 = m_1 \land m_2$, a contradiction). Denote by $C$ the $\{R\}$-constraint $C = R(x_1, \dots, x_k)$. Set $M_1(u,x,y,v) = C[V_{0,0}/u, V_{0,1}/x, V_{1,0}/y, V_{1,1}/v]$. It contains $\{0011, 0101\}$ (since $m_1, m_2 \in R$) but it does not contain $0001$ (since $m_1 \land m_2 \notin R$).

Let $R$ be a relation that is not dualHorn. Then there are $m_3,m_4 \in R$ such that $m_3 \lor m_4 \notin R$. For $i,j \in \{0,1\}$, set $V'_{i,j} = \{x \mid x \in V,\ m_3(x) = i,\ m_4(x) = j\}$. Observe that the sets $V'_{0,1}$ and $V'_{1,0}$ are nonempty (otherwise $m_3 = m_3 \lor m_4$ or $m_4 = m_3 \lor m_4$, a contradiction). Set $M_2(u,x,y,v) = C[V'_{0,0}/u, V'_{0,1}/x, V'_{1,0}/y, V'_{1,1}/v]$. It contains $\{0011, 0101\}$ (since $m_3, m_4 \in R$) but it does not contain $0111$ (since $m_3 \lor m_4 \notin R$). Finally consider the $\{R, (t \land \neg f)\}$-formula
$$M(f,x,y,t) = M_1(f,x,y,t) \land M_2(f,x,y,t) \land (t \land \neg f).$$
One verifies that it is equivalent to $(x\neq y) \land (t \land \neg f)$. Due to Prop.~\ref{prop:express_equality_literature}, 2nd item, $(t \land \neg f)$ is expressible as an $S$-formula, and therefore so is $M(f,x,y,t)$.
Since $\exists t,f M(f,x,y,t)$ is equivalent to $(x\neq y)$, we obtain $(x\neq y) \in \closneq{\csl}$.
\item[Horn.]
Let $R$ be a relation that is not essentially negative, but Horn. Then there are $m_1,m_2,m_3 \in R$ such that $m_4 := m_1 \land (m_2 \lor \neg m_3) \notin R$. Since $R$ is Horn, $m_5 := m_1 \land m_2 \in R$. For $i,j,k \in \{0,1\}$, set $V_{i,j,k} = \{x \mid x \in V,\ m_1(x) = i,\ m_2(x) = j,\ m_3(x) = k\}$. Observe that the sets $V_{1,0,0}$ and $V_{1,0,1}$ are nonempty (otherwise $m_5 = m_4$ or $m_1 = m_4$, a contradiction). Denote by $C$ the $\{R\}$-constraint $C = R(x_1, \dots, x_k)$.
Set $M(f,x,y,t) = C[V_{0,0,0}/f, V_{0,0,1}/f, V_{0,1,0}/f, V_{0,1,1}/f, V_{1,0,0}/x, V_{1,0,1}/y, V_{1,1,0}/t, V_{1,1,1}/t]$.

It contains $\{00001111,00000011\}$ (since $m_1, m_5 \in R$) but it does not contain $00001011$ (since $m_4 \notin R$). Finally consider the $\{R, (t\land \neg f)$-formula
$$M'(f,x,y,t) = M(f,x,y,t) \land M(f,y,x,t) \land (t \land \neg f)$$
One verifies that it contains $\{0111, 0001\}$ but not $0101$ and neither $0011$. Therefore it is equivalent to $(x=y) \land (t \land \neg f)$. Due to Prop.~\ref{prop:express_equality_literature}, 2nd item, $(t \land \neg f)$ is expressible as an $S$-formula, and therefore so is $M'(f,x,y,t)$. Since $\exists t,f M'(f,x,y,t)$ is equivalent to $(x = y)$, we obtain $(x = y) \in \closneq{\csl}$.
\item[dualHorn.] Analogously to the Horn case, using the property that $\csl$ is not essentially positive, but dualHorn.\qedhere
\end{description}
\end{description}

\end{proof}

\begin{restatelemma}[lem:independence]
\begin{lemma}
Let $\csl, \csl'$ be two constraint languages such that $\csl'$ is neither essentially positive nor essentially negative. 
Let $\abds{}{} \in \{\abd{}{}, \abdeq{}{}, \abdle{}{}\}$. If $\csl \subseteq \clos{\csl'}$, then
$\abds{\csl}{} \preduction \abds{\csl'}{}$. 
\end{lemma}
\end{restatelemma}
\begin{proof}
We may consider $\KB$ as a single $\csl$-formula.
We will transform $\KB$ into a corresponding $\csl'$-formula by replacing every $\csl$-constraint by a corresponding $\csl'$-formula. For this purpose we first construct a look-up table mapping any $R \in \csl$ to an $\csl'$-formula $F_R$ as follows.
Since $R \in \csl \subseteq \clos{\csl'}$, and by Lemma~\ref{lem:express_equality} $\clos{\csl'} = \closneq{\csl'}$, we have that $R \in \closneq{\csl'}$.
Thus, we have by definition an $\csl'$-formula $\phi$ such that $R(x_1, \dots, x_n) \equiv \exists y_1 \dots \exists y_m \phi(x_1, \dots, x_n,y_1, \dots, y_m)$, where we can assume the $x_i$'s and $y_i$'s to be $n+m$ distinct variables.
We obtain $F_R$ by removing the existential quantifiers.
Note that the computation of the look-up table takes constant time, since $S$ is finite and not dependent on any input. We are now ready to transform $\KB$ into an appropriate $\csl'$-formula by applying the following replacement procedure as long as applicable.

\begin{itemize}
\item Be $C_R = R(x_1, \dots, x_n)$ an $\csl$-constraint (now the $x_i$'s are not necessarily $n$ distinct variables). Replace $C_R$ by its corresponding $\csl'$-formula $F_R(x_1, \dots, x_n, y_1, \dots, y_m)$, where the variables $y_1, \dots, y_m$ are fresh variables that are unique to $C_R$ (they will not be used for any other constraint replacement).
\end{itemize}

This transformation procedure introduces additional variables. We show that their total number is polynomially bounded. Denote by $m_R$ the number of $y_i$'s added while replacing $C_R$ (denoted $m$ in the above procedure). One observes that the total number of additional variables is bounded by the number of original constraints times the maximum of all $m_R$.
Since $m_R$ is only dependent on $R$, it is constant. Since $S$ is finite, the maximum of all $m_R$ is constant. We conclude that the transformation can be achieved in polynomial time.

It remains to observe that the so obtained abduction instance has exactly the same solutions as the original instance.
\end{proof}

\begin{restatelemma}[lem:base-ind-ess-positive]
\begin{lemma}
Let $\csl, \csl'$ be two constraint languages such that $\csl'$ is essentially positive. 
Let $\abds{}{} \in \{\abd{}{}, \abdeq{}{}, \abdle{}{}\}$. 
If $\csl \subseteq \clos{\csl'}$, then $\abds{\csl}{} \preduction \abds{\csl'}{}$. 
\end{lemma}	
\end{restatelemma}
\begin{proof}
	The general case is due to Nordh and Zanuttini~\cite[Lemma~22]{DBLP:journals/ai/NordhZ08}.
	The result for `$\leq$' is because of the following.
	Removing equality clauses and deleting the duplicating occurrences of variables only decreases the size of $H$. 
	
	Now we proceed with proving the case for `$=$'.
	We show that for any $\abdeq{\csl \cup \{=\}}{}$ instance $\abdformat{V, H, M, \KB}{}$, there is an $\abdeq{\csl}{}$-instance $\abdformat{V_1, H_1, M_1, \KB_1}{}$ such that the former has an explanation if and only if the later has one.
	The proof uses the fact that the only negative clauses in $\KB$ are of size $1$.
	Since the existence of a solution is invariant under the equality clauses, we only need to assure that the size of a solution is also preserved. 
	For each clause $x_i= x_j \in \KB$, we do the following:
	 \begin{enumerate}
	 	\item If at most one of the $x_i,x_j$ appears in $H$, remove the clause $x_i = x_j$ from $\KB$, replace $x_j$ by $x_i$ everywhere in $\KB \cup H\cup V \cup M$ (and delete $x_j$).
	 	\item If both $x_i, x_j$ are from $H$. Then,
	 		\begin{itemize}
		 		\item if $\neg x_i$ (resp., $\neg x_j$) appears in $\KB$, we add $\neg x_j$ (resp., $\neg x_i$) to $\KB$ and remove the clause $x_i = x_j$ from $\KB$.
		 		\item otherwise $\neg x_i,\neg x_j\notin\KB$ and then simply remove the clause $x_i = x_j$ from $\KB$ and do not remove any variable.
		 	\end{itemize}
	 \end{enumerate}
%Recall that in the absence of a condition on size of solution (i.e., $\abd{\csl}{k}$), the result already holds due to Galois correspondence of the languages.

The problem caused by equality clauses is the following.
If we remove a variable which is also a hypothesis, then removing this from $H$, owing to some equality constraint, may not preserve the size of solutions.
Furthermore, this problem occurs only when an equality clause contains both variables from $H$ (case~2.), since otherwise the size of $H$ is not changed (case~1.).
% Note that even though we do get rid of the equality clause, the size of $H$ is not changed in second case because we do not remove anything from it.
We prove the following correspondence between the solutions of the two instances.
\begin{claim}
	A subset $E\subseteq H$ is an explanation for $\abdeq{\csl \cup \{=\}}{}$ if and only if $E$ is an explanation for $\abdeq{\csl}{}$.
\end{claim}
\begin{claimproof}
	``$\implies$'': %This direction is trivial. 
	If $E\land \KB$ is consistent, we prove that $E\land \KB_1$ is also consistent.
	Note that $\KB_1 \subseteq \KB$, except if $\neg x_j\in \KB_1$ for some $x_j$. 
	This implies that $x_j\not \in E$ due to the reason that $x_i = x_j$ and $\neg x_i \in \KB$.
	Finally, $M_1\subseteq M$ and $E$ is an explanation for $M$ implies $E$ is also an explanation for $M_1$.
	
	``$\impliedby$'':
	Suppose that $E \land \KB_1$ is consistent and let $\theta$ be a satisfying assignment.
	%entails $M_1$.
	We consider each equality constraint separately and prove that $E \land \KB$ is consistent.
	In the first case, for each $x_i, x_j$ such that at most one (say $x_i$) appears in $H$.
	If $x_i \in E \subseteq H$ then $\neg x_j \not \in\KB_1$ since this would imply that $\neg x_i \in \KB_1$ and $x_i \not \in E$. 
	Consequently, $E \land \KB_1 \land (x_i=x_j)$ is consistent (by extending $\theta$ to $\theta(x_i)= 1 =\theta(x_j)$).
	On the other hand, if $\neg x_j  \in\KB_1$ then $\neg x_i  \in\KB_1$ and $x_i \not \in E$.
	As a result, $\theta$ extended to $\theta(x_i)= 0 =\theta(x_j)$ is a satisfying assignment.
	In the second case, if both $x_i, x_j \in H$ then we also have two sub-cases based on whether $\neg x_i \in \KB_1$ or not.
	If $\neg x_i \in \KB_1$ then due to case~2, we have $\neg x_j\in \KB_1$ and this implies $x_i\not \in E$.
	As a consequence, $E \land \KB_1 \land (x_i=x_j)$ is consistent by extending $\theta$ to $\theta(x_i)= 0 =\theta(x_j)$. 
	In the sub-case when both $x_i, x_j \in H$ and $\neg x_j \not \in \KB_1$ then mapping $\theta(x_i)=1 = \theta(x_j)$ satisfies $E\land \KB_1 \land (x_i= x_j)$.
	This is because all the non-unit clauses in $\KB_1$ are positive. 
	Since this is true for all the equality clauses it shows that  $E\land\KB$ is consistent.
	
	For entailment, if for some $m_i, m_j \in M$ $(m_i = m_j) \in \KB$ and $m_j \not \in M_1$, that is $M_1\subsetneq M$.
	Since $E\land \KB_1$ is consistent and entails $M_1$, we have $E\land \KB_1 \land (m_i=m_j)$ is also consistent (due to arguments for consistency) and entails $M_1 \cup \{m_j\}$.
	This completes the proof in this direction and settles the claim.
\end{claimproof}

Finally, the above reduction can be computed in polynomial time because both steps are applied once for each equality clause and each step takes polynomial time.
This shows the desired reduction between $\abdeq{\csl \cup \{=\}}{}$ and $\abdeq{\csl}{}$.
\end{proof}

\subsection{Fixed-parameter tractable cases}
\begin{restatelemma}[lem:ABD-IS02]
\begin{lemma}
	The classical problems \abdeq{\csl}{} and \abdle{\csl}{} are in \P for $\clos\csl\subseteq\IS{}{02}$.
\end{lemma}
\end{restatelemma}
\begin{proof}
We start with $\abdeq{\csl}{}$.
Denote by $\KB', H',M'$ the result of applying the unit propagation on those literals $y$ such that $y \in \text{Lit}(\KB) \backslash (H^{+} \cup M^{-})$.
Recall that for a set $Y$ of literals, $Y^{+}$ (resp., $Y^{-}$) denotes the set of positive (negative) literals formed upon $Y$. 
In unit propagation, for a unit clause $u$, any clause containing $u$ can be deleted and delete in any clause $\sim\!u$, where $\sim\!u=x$ if $u=\lnot x$ is a negative literal and $\sim\!u=\lnot x$ if $u=x$ is a positive literal.
Note that literals $y \in H^{+}\cup M^{-}$ (that is, $y\in H$ or $y=\neg m$ with $ m \in M$) are excluded from this rule as mentioned above.
The reason for this choice is as follows. 
If $\neg m \in \KB$ for some $m\in M$ then removing $m$ from $\KB\cup M$ transforms a `no solution'- to a `yes solution'-instance.
Similarly, removing an $h\in H$ from $\KB\cup H$ may decrease the solution size of the instance. 
Finally, the positive literal $m\in M$ may or may not be processed.
However, it is important to consider $h\in H^{-}$ since this helps in invalidating the clauses of length $\geq 2$. 
%We can also allow some positive hypotheses to propagate provided that those appear in M as well. Thereby, we already get rid of those manifestations which are already explained by $\KB$. 
%\todo[inline]{I realised that this condition is necessary, otherwise we may have to delete the 'propagated' variables from M and H as well. Which does not preserve the solutions. E.g. take KB be $ x_1\lor m, \neg x_1$ with $m\in M$. $m\in KB$ as a positive literal should not be propagated. For H, the removal of positive literals may cause problem when extending a solution, removal of negative ones is fine, (maybe necessary, but this can be taken care of as we do below!)} 

Let $P$ and $N$ be the positive, respectively negative unit clauses of $\KB'$ over $\clos\csl$. 
%\todo{P now corresponds to H and N to M}
Note that if $N\not =\emptyset$ then there can be no explanation for $M$. 
This is due to the fact that only negative unprocessed literals are over $M$ implying that $\KB$ is inconsistent with $M$.
Because of this, we have $N=\emptyset$. 
Moreover, the positive clauses of length $\geq 2$ in $\KB'$ do not explain anything as a variable cannot be enforced $0$.
Therefore, a positive literal $x$ cannot explain anything more than $x$ itself.
This implies that there is an explanation for $M$ if and only if $M' \subseteq H' \cup P$.
% Because unit positive clauses has already been removed from both M and KB... therefor $M \subseteq P \cup H$ is not required.
Now, the set $M'\setminus P$ denotes those $m\in M$ which are not already explained by $\KB$ and must be explained by $H'$. 
As a consequence, there exists an explanation for $\abdle{\csl}{}$ if and only if $M'\setminus P \subseteq H'$ and $|M'\setminus P| \leq s$. 
The consistency is already assured by the fact that $N=\emptyset$.
Finally, to determine whether there is an explanation $E \subseteq H$ of size $s$, it suffices to check additionally whether $|H'| \geq s$.
This argument ensures whether we can artificially increase the solution size, since, in that case an $E \subseteq H'$ with above conditions constitutes an explanation for the problem \abdeq{\csl}{}.
If this is not true, then no explanation of size $s$ exists.
The unit propagation and size comparisons can be done in polynomial time, which proves the claim.
Finally, the result for $\abdle{\csl}{}$ is due to Lemma~\ref{monotone}.
\end{proof}

\begin{restatelemma}[lem:ABD-E-D1]
\begin{lemma}
	The classical problems $\abdeq{\csl}{}$ and $\abdle{\csl}{}$ are in \P if $\clos\csl\subseteq\ID_1$. %[2-affine]
\end{lemma}	
\end{restatelemma}
\begin{proof}
Analogously to Creignou et~al.~\cite[Prop.~1]{DBLP:conf/sat/CreignouOS11}, we change the representation of the $\KB$. 
Without loss of generality, suppose $\KB$ is satisfiable and contains no unit clauses since unit clauses can be dealt with in a straightforward way. 
Each clause expresses either equality or inequality between two variables. With the transitivity of the equality relation and the fact that (in the Boolean case) $a \neq b \neq c$ implies $a = c$, we can identify equivalence classes of variables such that each two classes are either independent or they must have contrary truth values.
We call a pair of dependent equivalence classes $(X, Y)$ a \emph{cluster} ($X$ and $Y$ must take contrary truth values).
Denote by $X_1, \dots, X_p$ the equivalence classes that contain variables from $M$ such that $X_i \cap M \neq \emptyset$. 
Denote by $Y_1, \dots, Y_p$ the equivalence classes such that for each $i$ the pair $(X_i, Y_i)$ represents a cluster.
We make the following stepwise observations.
\begin{enumerate}
\item There is an explanation iff $\forall i: H \cap X_i \neq \emptyset$.
\item The size of a minimal explanation ($E_{min}$) is $p$, it is constructed by taking exactly one representative from each $X_i$.
\item There exists an explanation of size $\leq s$ iff $p \leq s$.
\item An explanation of maximal size ($E_{max}$) can be constructed as follows:
	\begin{enumerate}
	\item $E := \emptyset$,
	\item for each $i$ add to $E$ all variables from $X_i \cap H$,
	\item for each cluster $(X,Y) \notin \{(X_i, Y_i) \mid 1 \leq i \leq p\}$:
		\begin{enumerate}
		\item if $|X\cap H| \geq |Y\cap H|$: add to $E$ the set $X\cap H$,
		\item else: add to $E$ the set $Y\cap H$.
		\end{enumerate}
	\end{enumerate}
\item Any explanation size between $|E_{min}|$ and $|E_{max}|$ can be constructed.
\item There is an explanation of size $=s$ iff $|E_{min}| \leq s \leq |E_{max}|$.\qedhere
\end{enumerate}
%\todo[inline]{Johannes: todo: finish teXing the observations}
\end{proof}

\subsection{Parameter `number of explanations' |E|}
\subsubsection*{Intractable cases}
\begin{lemma}\label{lem:ABD-E-IE}
$\abdeq{\csl}{|E|}$ and $\abdle{\csl}{|E|}$ are $\W\P$-complete if $\,\IE \subseteq \clos{\csl} \subseteq \IE_2$.
\end{lemma}
\begin{proof}
For $\IE_2$ we have $\W\P$-membership from the argument that $\SAT(\IE_2)$ and $\IMP(\IE_2)$ are in $\P$ (cf. Prop.~\ref{prop:SAT-IMP-P}).
The hardness proof from \cite[Cor.~9]{DBLP:conf/aaai/FellowsPRR12} for definite Horn theories ($\IE_1$) works for $\IE$ as well. 
The only types of clauses used are $x \land y \rightarrow z$ and $x \rightarrow y$, which are both in $\IE$ and consequently expressible by $\csl$ as $\IE \subseteq \clos{\csl}$. 
Both, membership and hardness arguments are valid for $\abdle{\csl}{|E|}$ as well (the problem in \cite[Cor.~9]{DBLP:conf/aaai/FellowsPRR12} used for hardness is monotone circuit sat, which is \emph{monotone}).
\end{proof}

%\begin{lemma}\label{lem:ABD-E-IL3}
% 	$\abdle{\csl}{|E|}$ and $\abdeq{\csl}{|E|}$ are in $\W\P$ if $\clos\csl\subseteq\IL_2$.
%\end{lemma}
%\begin{proof}
%	$\SAT(\IL_2), \IMP(\IL_2)\in\P$ (Prop.~\ref{prop:SAT-IMP-P}).
%	We need to find an explanation $E$ such that $\KB\land E$ is satisfiable and $\KB\land E\models M$.
%	There are ${{|H|}\choose{|E|}}\in O(n^k)$ many candidates which we can address with $k\cdot\log(n)$ many bits where $k=|E|$.
%	Guessing these bits proves $\W\P$-membership.
%\end{proof}

\begin{lemma}\label{lem:ABD-E-IM-W2}
	\abdeq{\csl}{|E|} is $\W{2}$-complete if $\clos{\csl} =\IM$.
\end{lemma}
\begin{proof}
	For membership we prove that $\abdeq{\IM}{|E|}\fptreduction \wsat{2}{1}$. 
	The latter is known to be $\W{2}$-complete (Proposition~\ref{theorem-wsat}). 
	Let $\pabdformat{V,H,M,\KB,s}{s}$ be an instance of $\abdeq{\IM}{|E|}$, where the solution size is the parameter.
	Specifically, let $\KB=\bigwedge\limits_{i\leq r}(x_i \rightarrow y_i)$ and $M= m_1\land \ldots \land m_{|M|}$.
	Note that, in order to explain a single $m_i\in M$, a single $h\in H$ suffices. 
	As a result, for each $m_i\in M$ we associate a set $H_i \subseteq H$ of hypotheses that explains $m_i$. 
	This implies that every element (singleton subset) of $H_i$ explains $m_i$. 
	Now, it is enough to check that at least one such $h\in H_i$ can be selected for each $m_i$. 	 
	For this we map $\pabdformat{V,H,M,\KB,s}{s}$ to $(\phi,k)$ where $\phi = \bigwedge  \limits_{i\leq m} \bigvee\limits_{x\in H_i} x$. 
	 Then our claim is that $\pabdformat{V,H,M,\KB,s}{s}$ has an explanation $E$ if and only if $\phi$ has a satisfying assignment of size $k$. 
	 Clearly, there is a $1\text-1$-correspondence between solutions $E$ of $\pabdformat{V,H,M,\KB,s}{s}$ and satisfying assignments $\theta $ with weight $k$ for $\phi$. 
	 That is, $\theta(x)=1 \iff x\in E$.  
	 
%\end{proof}
%\begin{lemma}\label{lem:ABD-E-IM-W2}
%	\abdeq{\IM}{|E|} is $\W{2}$-hard.
%\end{lemma}
%\begin{proof}
	
	For hardness, we reduce from $\wsatpos{2}$ which is $\W{2}$-complete by Proposition~\ref{theorem-wsat}.
	Given $\bigwedge\limits_{i\in q}\bigvee\limits_{j\in  r}(X_{ij}^+) = \bigwedge\limits_{i\in  q}(X_{i_1}\lor \ldots \lor X_{i_r})$, where $\var(\alpha)=\setdefinition{X_{ij}\mid i\in q, j\in r} $, we let
	$\KB= \bigwedge\limits_{i\in q}\bigwedge\limits_{x\in h_i}(x\rightarrow h_i)$, 
	%
	%\todo{Should we have $x_{i_j}$ here? Johannes: NO :)}
	%
	$H=\var(\alpha)$, $M= \bigwedge\limits_{i\in q}h_i $ and $V=H\cup M$.	
	Then for a subset $E\subseteq H$ we have that, $E$ is an explanation for $\abdeq{\IM}{|E|} \iff \theta \models \phi$ where $\theta(x)=1 \iff x\in E$.
	%
%	\textbf{Reduction from \hset.} Given $\hset$ instance $\langle V,\tilde E,k\rangle$ s.t. $V=\{v_1,\ldots,v_n\}$ and $\tilde E=\{e_1,\ldots,e_m\}$ where each $e_i\subseteq V$ then letting $ H=V$,  $Var= V\cup  M \text{ and } M= h_1\land\ldots\land h_m $
%	$$\KB=\bigwedge\limits_{e_i\in \tilde E} \bigwedge\limits_{v\in e_i}(v\rightarrow h_i)$$
%	(If there is a hitting set $HS$ of size $k$ then it must contain at least one element from each edge $e_i$ whence it explains the manifestation $ h_i$. Also, $HS \land T$ is consistent by definition. Conversely, any explanation $Exp$ of abduction problem of size $k$ must be a hitting set since $Exp\land T \models M $ iff for each edge $e_i$, there is a $v\in Exp$ s.t. $v\rightarrow h_i \in T$ and by assumption $T\land Exp$ is consistent.)
\end{proof}
%\begin{corollary}\label{cor:ABD-M-IM}
	 %The problems \abdle{\IM}{|E|} and  \abdeq{\IM}{|E|} are $\W{2}$-complete.
%\end{corollary}
%\begin{proof}
	%Follow from \mathref{Lemma }{monotone} and above lemmas.
%\end{proof}
%\begin{lemma}\label{lem:ABD-E-IS00}
	%The problems $\abdeq{\IS{}{00}}{|E|}$ and $\abdle{\IS{}{00}}{|E|}$ are $\W{2}$-complete.
%\end{lemma}
%\begin{proof}
	%Same proof as $\abdeq{\IM}{|E|}$. The difference between the three case arises because of the fact that if in general abduction problem, a solution is found of size $\geq s$ then finding one of size $=s$ or $\leq s$ is difficult. (Equivalent to finding minimal solution).
%\end{proof}
We prove that the $\W{2}$ membership from previous lemma can be extended to  $\IS{}{10}$
\begin{lemma}\label{lem:ABD-E-IS10-inW2}
	Let $\ell\ge2$, then \abdeq{\csl}{|E|} is in \W{2} if $\clos\csl\subseteq\IS{\ell}{10}$. %[IHS-B-]
\end{lemma}
\begin{proof}
	We reduce our problem to $\wsat{2}{\ell}$ which is $\W2$-complete due to Proposition~\ref{theorem-wsat}.
	Consider the reduction from Lemma~\ref{lem:ABD-E-IM-W2} again, where we map $\pabdformat{V,H,M,\KB,s}{s}$ to $(\phi,k)$,	where $\phi = \bigwedge  \limits_{i\leq m} \bigvee\limits_{x\in H_i} x$.
	The only difference from Lemma~\ref{lem:ABD-E-IM-W2} is that in $\IS{\ell}{10}$ there are additional constraints of the form $(\neg x_1 \lor \ldots \lor \neg x_q)$ where $q\leq \ell$. 
	  Now we have two cases.
	
If all the additional constraints contain exclusively variables from H then we simply add these constraints to $\phi$ and obtain a new formula $\psi$. 
Since any satisfying assignment for $\psi$ would satisfy these constraints as well as $\phi$ and therefore, is an explanation as required. 
Conversely, any explanation would yield a satisfying assignment for this new formula $\psi$ since this explanation is consistent with $\KB$.

Now suppose that constraints contain variables that are not from $H$. 
We transform such constraints into their equivalents which contain variables only from H. 
To achieve this we repeat the following procedure as long as applicable:

Pick a variable $u \not \in H$ occurring in a constraint $C_u$. 
Compute the set of hypotheses $H_u \subseteq  H$ that explain $u$ 
(analogously to Lemma~\ref{lem:ABD-E-IM-W2}). 
Let $H_u = \setdefinition{h_1,\ldots, h_r}$. 
Now we replace the constraint $C_u$ by $r$ copies of itself and in each $C_u^i$ we replace the variable $u$ by $h_i$.
Note that this does not change the width of any clause. 
Finally, we add these clauses to $\phi$ and obtain a new formula $\psi$. 
\begin{claim} 
	The above construction preserves the correspondence between the solutions of $\abdeq{\IS{\ell}{10}}{|E|} $ and the satisfying assignments of $\phi$ with weight $k$. Moreover, it can be achieved in polynomial time.
\end{claim}
\begin{claimproof}
	Note that the difference between Lemma $\ref{lem:ABD-E-IM-W2}$ and this case is in the fact that a solution to $\abdeq{\IS{\ell}{10}}{|E|}$ must satisfy additional constraints as specified above. 
	The problematic part is when some variables $x_i,\ldots x_j$ are in $H$ and some constraint over these variables appears in the $\KB$. 
	The formula $\psi$ must not allow such elements to be the part of solution since the constraints stop from certain elements to appear together in the solution (being negative clauses). 
	This proves the first claim in conjunction with the arguments in Lemma~\ref{lem:ABD-E-IM-W2}.
	
	Now we prove that this transformation works in polynomial time. 
	The worst case is when a clause contains no variable from H. 
	Furthermore assume that this clause is of maximum arity, say $C= (\neg x_1,\ldots \neg x_q)$ where $q\leq \ell$ and $q$ is the maximum arity of constraint language in $\KB$. 
	 Each $x_i$ can have the associated set $H_{x_i}$ of maximum size $n$ where $n$ is input size. 
	 Hence each clause will be blown-up to atmost $n^q$ new constraints at the completion of the above procedure. 
	 As $q$ is constant (only depends on the constraint language and not on the input), the factor $n^q$ is polynomial. 
	 Since there are polynomial many constraints to check for this procedure, we conclude that the transformation takes only polynomial time.
\end{claimproof}
Eventually, similar arguments as in Lemma~\ref{lem:ABD-E-IM-W2} for $\psi$ complete the proof.
\end{proof}
%\todo[inline]{My thoughts on the current issue are. (1) Either we switch our reductions to MC-Sigma-Formulas (which is time consuming).
%	(2) Just write that since $ABD_=(IS_{10}, |E|)$ is $\W2$-complete, it is reducible to the problem MC($\Sigma_{2,d}$). 
%	Thus -blah blah- we have a structure $\mathcal A $ and a $\Sigma_{2,d}$ formula $\phi$ such that $ABD_=(IS_{10}, |E|)$ has a solution iff $\mathcal{A} \models \phi$. 
%	Obviously in $\phi$ there is a subformula that assure solution of size $k$ by saying k-distinct elements. 
%	Then replacing this formula $\phi$ with $\phi'$ (removing k-distinct elements condition) yields \W2 membership.
%	Alternatively, we can only mention this in conclusion that we believe this works thus giving \W2 upper bound. 
%	(3) \emph{Not sure if this really works:} Since finding a solution of size $k$ is in \W2. In order to determine whether there is a solution of size at most $k$ we determine whether there is a solution of size $1$ of size $2$ so on, of size $k$. That is fpt-many repetitions of a ``$\W2$-algorithm''. Which should give \W2 membership.}
	
\begin{lemma}\label{lem:ABD-E-IV2-inW2}
	\abdeq{\csl}{|E|} is in \W{2} if $\clos\csl\subseteq\IV_2$. %[dualHorn]
\end{lemma}
\begin{proof}
We extend the proof for \W{2}-membership for the case $\IM$ (see Lemma~\ref{lem:ABD-E-IM-W2}).
In the \IM-case, we dealt only with clauses of type $x \rightarrow y$. 
We refer to such classes as type-$0$ clauses.
In $\IV_2$ we have additional clauses of the following types:
\begin{enumerate}
\item Unit clauses: both positive and negative. $x$, $\neg x$
\item Positive clauses of size two or greater: $(x_1 \lor \dots \lor x_n)$, $n\geq 2$
\item Clauses with exactly one negative literal of size $3$ or greater: $(\neg x_0 \lor x_1 \lor \dots \lor x_n)$, $n\geq 2$
\end{enumerate}
We can we eliminate the type-$1$ clauses by unit propagation and obtain thereby a satisfiability equivalent formula. 
%That is, for a positive unit clause $x$, any clause where $x$ occurs can be deleted and in any clause where $\overline{x}$ occurs, the $\overline{x}$ can be deleted. 
%Consequently, every unit clause $x$ itself can be deleted.
%Proceed analogously with negative unit clauses.
Note that this transformation process can generate additional clauses of type-$0$, type-$2$, or type-$3$. 
As a consequence, we end up only with clauses of either type-$0$, type-$2$, or type-$3$ and, particularly, no type-$1$ clauses anymore. 
This transformation does not preserve all the satisfying assignments but those can be maintained by adding fixed values of the eliminated variables to the assignment.

Now, we argue that by applying resolution on the variables in $\KB \setminus H$, we can ignore type-2 and type-3 clauses.
Notice that to the variables in $H$, we do not apply resolution. 
Recall the idea behind the construction of Lemma~\ref{lem:ABD-E-IM-W2}, we want to come up with a formula which has a satisfying assignment if and only if our abduction instance has an explanation.
 A satisfying assignment that selects a variable $x \in H$ (maps $x$ to $1$) forces all the variables $y_1,\ldots, y_n$ such that $(x\rightarrow y_i) \in \KB$ to be mapped $1$ for $i\leq n$. 
Furthermore, it also forces each $z_{i,j}$ such that $(y_i \rightarrow z_{i,j}) \in \KB$ to be mapped $1$, and so on. 
This precisely captures the intuition that $x$ (as a hypothesis) explains each $y_i$ and $z_{i,j}$.  
As a consequence, removing such variables from $H$ (owing to resolution) in the case when those variables explain some manifestation would be problematic.
%For example, consider $\KB=\setdefinition{\neg h \lor y_1, h\lor x_1\lor x_2, y_1\lor m}$ with $H=\{h\}$ and $M=\{m\}$.
%Then applying resolution on all the variables results in the clause \{$x_1\lor x_2\lor m\}$, whose satisfaction using variables from $H$ does not guarantee an explanation while applying resolution not on variables in $H$ results in $\neg h\lor m, h\lor x_1\lor x_2$ which does guarantee (namely, taking any map such that $h\mapsto 1$ ).

Finally, we prove the claim that for \abdeq{\IV_2}{|E|} we can ignore the type-$2$ and type-$3$ clauses.
Type-2 clauses are irrelevant since the satisfaction of such clauses does not force any particular variable to $1$. 
In a type-3 clause ($C= \neg x_1 \lor x_2 \lor \dots \lor  x_m$) the variable $x_1$ forces a whole clause to be true (at least one of the remaining variables must be mapped to $1$).
Such clauses cannot be ignored right-away because there might be further clauses of the form $\neg x_j \lor m$ for each $2\leq j \leq m $ with $m\in M$ and an explanation to $m$ might be lost (selecting $x_1$ in the solution). 
However, after applying resolution we know that type-3 clauses only force one of the many positive variables to $1$ and do not actually force a single variable to $1$. 
As a result, this allows us to ignore type-3 clauses as well.
% Because if any such clause was indeed "useful", after resolution it would be reduced to a type-0 clause, and If not, then it does not explain anything (m or h). 
%
Consequently, we are only left with type-0 clauses. 
This completes the proof by the same arguments as in the proof of Lemma~\ref{lem:ABD-E-IM-W2}.
\end{proof}

\begin{lemma}\label{lem:ABD-eq-E-IS1h2-hardness}
	For any constraint language $\csl$ such that $\neg x\lor \neg y \in \closneq\csl$, the problem $\abdeq{\csl}{|E|}$ is \W1-hard.% [negative of width 2]
\end{lemma}
\begin{proof}
The problem \IndSet is known to be $\W1$-hard \cite[Thm.~10.8]{DBLP:series/mcs/DowneyF99}.
We reduce \IndSet to $\abdeq{\csl}{|E|}$.
Let $((V, \tilde E), k)$ be an instance of $p$-\IndSet and $k$ the parameter. 
We map it to $\pabdformat{V,H,M,\KB,k+1}{k+1}$, where
\begin{align*}
\KB & := \{(\neg x \lor \neg y) \mid (x,y) \in \tilde E\},\\
H & := \var(\KB) \cup \{z\},\\
M & := z. 
\end{align*}
Let $U$ be an independent set of size $k$ then $U\land \KB$ is consistent because no two elements with an edge are in $U$.
As a consequence, $U\cup \{z\}$ is an explanation for $\pabdformat{V,H,M,\KB,k+1}{k+1}$.
Conversely, an explanation $E$ for $\pabdformat{V,H,M,\KB,k+1}{k+1}$ of size $k+1$ must include $z$ as well as $k$ other variables. 
Now, $E \land \KB$ is consistent and this implies that no variables in $E$ have an edge, consequently giving an independent set of size $k$.
This implies that $(V, \tilde E)$ admits an independent set of size $k$ if and only if $\abdinstance$ admits an explanation of size $k+1$.
\end{proof}
%\todo[inline]{I have the feeling that the lemma below works with equality constraints as well. Because |E| is the parameter and this equals "s" now. In the reduction $\abdeq{S\cup \{=\}}{|E|}\fptreduction \abdeq{S}{|E'|}$,  $|E|$ and $|E'|$ are also different. Proof idea: (solutions are already invariant under = we only assure the size-invariance) For $\implies$. From E,  an explanation for RHS, remove those variables that appear in an equality-clause. Remember that you can only remove those variables from H that appear in equality clauses in KB thus reducing the size of H'. This also effects the new parameter E' which is now smaller. Similar argument works for $\impliedby$, since from E' a solution for RHS, we can add more elements to reach the size |E| by adding those variables from H to E' that are "equality-neighbours" of the variables already in E'.} %In other words, we remove equality constraints as usual and from the information about sizes of E and E' we can construct those which caused H to be smaller 
\begin{lemma}\label{lem:ABD-eq-E-IS12hk}
	Let $\ell\ge 2$, then $\abdeq{\csl}{|E|}$ is in $\W1$ if $\clos\csl\subseteq\IS{\ell}{12}$.
\end{lemma}
\begin{proof}
We reduce $\abdeq{\csl}{|E|}$ to $\wsat{1}{\ell}$, which is $\W1$-complete (Prop.~\ref{theorem-wsat}).
Note that $\Gamma_{1,\ell}$ is the class of $\ell$-CNF formulas.
We want to mention here that the proof is correct even in the presence of equality constraints. 
As a consequence, the base independence is not implied by any of the previous lemmas but it follows due the proof below.

According to Lemma~\ref{lem:ABD-leq-E-IS12}, we can determine whether there exists a solution of size $\leq s$ in polynomial time. 
Let $\pabdformat{V,H,M,\KB,s}{s}$ be an instance of $\abdeq{\IS{\ell}{12}}{|E|}$ with $\KB = \bigwedge_{i\leq r} C_i \land N \land P \land E$, where $C_i = (\neg {x^i_1} \lor \dots \lor \neg {x^i_\ell})$, and $P, N$ denote the positive and negative unit clauses, respectively, and $E$ are the equality clauses.
Without loss of generality, assume that $\pabdformat{V,H,M,\KB,s}{s}$ admits a solution of size $\leq s$ (otherwise, map it to a negative dummy instance). 
Moreover, it follows from Lemma~\ref{lem:ABD-leq-E-IS12} that in this case any solution $E$ satisfies that $E_{MP} \subseteq E \subseteq H$. 
This implies $s \geq |E_{MP}|$. 
We also know from Lemma~\ref{lem:ABD-leq-E-IS12} that $E_{MP}$ is an explanation for $M$ and that both $E_{MP}$ and $M$ are consistent with all clauses in $\KB$.

The question now reduces to whether we can extend $E_{MP}$ to a solution of size $s$ by adding $s - |E_{MP}|$ variables from $H \setminus E_{MP}$? 
We show that this can be achieved and map $\pabdformat{V,H,M,\KB,s}{s}$ to $\langle \varphi, s - |E_{MP}|\rangle$, where $\varphi$ is obtained from $\KB$ by the following consecutive steps:
First we take care of equality clauses.
For each $x_i=x_j \in \KB$, such that $x_i, x_j \in H$, 
%if $x_i \in E_{MP}$ (resp. $x_j$) then add $x_j$ ($x_i$) to $E_{MP}$ .
%If, in this first case, $x_i \not\in E_{MP}$ then 
add to $\phi$ the clauses $(\neg x_i\lor x_j)$ and $( x_i\lor \neg x_j)$.
%These clauses are exempted to the following procedure.
We add these two clauses to $\phi$ ensuring that corresponding to each clause of the form $x_i=x_j$, either both $x_i, x_j$ are in the solution, or none is. 
\begin{enumerate}
\item Remove all clauses $C_i$ containing only variables not from $H$.
\item Remove all negative unit clauses $(\neg x)\in N$ such that $x \notin H$.
Note that after this step all remaining negative unit clauses are built upon variables from $H \setminus E_{MP}$ only.
\item For each clause $C_i$, denote by $X^i_H$ (resp., $X^i_{\overline{H}}$) the variables from $H$ (resp., not from $H$). 
Execute the following:
\begin{enumerate}
	  \item Remove $C_i$.
	  \item If $X^i_{\overline{H}} \subseteq P$:
		add to $\varphi$ the clause $(\neg x_1 \lor \dots \lor \neg x_p)$, where 
		$\{x_1, \dots , x_p\} = X^i_H \setminus E_{MP}$. 
		Otherwise nothing needs to be done as $X^i_{\overline{H}}\not \subseteq P$ is true.
		Then, for some variable $x\notin P$ we have that $\lnot x\lor \bigvee_{x_j\in X^i_H}\lnot x_j$ is satisfiable via setting $x$ to $0$ if all $x_j$ are mapped to~$1$.
\end{enumerate}

Note that after this step all remaining clauses $C_i$ are built upon variables from $H$ only.

\item Remove all positive unit clauses $(x) \in P$ such that $x \notin H \setminus E_{MP}$.
Note that after this step it holds that $\var(\varphi) = H$ and all remaining positive unit clauses are built upon variables from $H \setminus E_{MP}$ only.

\item For all clauses $C_i$: remove from $C_i$ all literals built upon variables from $E_{MP}$.
Note that in the so obtained $C_i'$ at least one literal remains, because otherwise $E_{MP}$ would be inconsistent with $C_i$.
\end{enumerate}
After the last step has been implemented, it holds that $\var(\varphi) = H \setminus E_{MP}$.
As a consequence, the following equivalences are true:
\begin{align*}
     & \pabdformat{V,H,M,\KB,s}{s} \text{ admits a solution of size exactly } s\\
\Leftrightarrow\;  & E_{MP} \text{ extends to a solution of size $s$ by adding $s - |E_{MP}|$ variables from } H\setminus E_{MP}\\
\Leftrightarrow\;  & \varphi \text{ has a satisfying assignment of size exactly } s - |E_{MP}|.\qedhere
\end{align*}
\end{proof}

\subsection{Parameter `number of manifestations' |M|}
\subsubsection*{Fixed-parameter tractable results}
\begin{lemma}\label{lem:ABD-M-IM-FPT}
	The problem \abdle{\csl}{|M|} is \FPT if $\clos\csl \subseteq\IM$.
\end{lemma}
\begin{proof}
	Given an instance $\pabdformat{V,H,M,\KB,s}{|M|}$ with $\KB=\bigwedge\limits_{i\leq r}(x_i \rightarrow y_i)$ and $M= m_1\land \ldots \land m_k$.
	Recall that each $m_i\in M$ can be explained by a singe $h_i\in H $.
	If $|M|\leq s$ then there is nothing to prove. 
	This is due the fact that in the proof of Lemma~\ref{lem:ABD-E-IM-W2}, there are fewer than $s$ many sets of the form $H_i$ each explaining an $m_i\in M$. 
	As a consequence, we need only select one $h_{i,j}$ from each $H_i$ as the part of a solution to yield a solution of size $\leq s$.
%	As a consequence, there is a solution if and only if there is a solution of size $\leq s$. 
	Accordingly, assume that $|M| > s$.
	Proceed as in the proof of Lemma~\ref{lem:ABD-E-IM-W2} and associate a set $H_i \subseteq H$ of hypotheses with each $m_i$ that explains it for $i\leq |M|$. 
	It is enough to check whether selecting at most $s$ many elements $h_i\in H$ can explain all the manifestations $m_i \in M$.
%	This problem differs from the quoted problem in the parametrisation we use.
	We reduce our problem to $\maxsat$ \cite{DBLP:journals/ita/BonnetPS16} (we alter the notation slightly)  asking, given a $\cnfFont{CNF}$ formula on $n$ variables with $m$ clauses, if setting at most $s$ variables to true satisfies at least $k$ clauses. 
	
	Let $H'$ be the collection of all $H_i$'s.
	For each $i$ let $C_i$ be the clause $\bigvee\limits_{j} h_j$ where $h_j\in H_i$.
	Furthermore, let $C$ be the collection of all such clauses.
	Then $C$ is built over variables in $V' = \bigcup_i H_i$. 
	Our reduction maps $\pabdformat{V,H,M,\KB,s}{|M|}$ to $\langle C, s, |M| \rangle $.
	Note that we only have $|M|=k$ many clauses in $C$ and, as a result, the question reduces to whether it is possible to set at most $s$ variables from $V'$ to satisfy every clause in $C$?
	The reduced problem $\maxsat$ when parametrised by $k$ (the minimum number of clauses to be satisfied) is $\FPT$ \cite[Prop.~4.3]{DBLP:journals/ita/BonnetPS16}.
	Now, each $H_i$ can be computed in polynomial time, the whole computation is a polynomial time reduction.
	Finally, the new parameter value $k$ is exactly the same as the old parameter $|M|$, the reduction is an $\fptreduction$-reduction.
	As a consequence, the lemma applies. 
%\todo{two sentences. one about time-bound the next parameter-bound}
	\end{proof}

\begin{corollary}\label{lem:ABD-M-IV2-FPT}
	The problem \abdle{\csl}{|M|} is \FPT if $\clos\csl \subseteq\IV_2$.
\end{corollary}
\begin{proof}
	We extend the $\FPT$-membership proof  Lemma~\ref{lem:ABD-M-IM-FPT} using the same argument as in Lemma~\ref{lem:ABD-E-IV2-inW2}.
	That is, after applying unit propagation and resolution we can ignore the positive clauses of length $\geq 2$ and clauses with one negative literal of length $\geq 3$.
\end{proof}

\begin{corollary}\label{cor:ABD-M-IM-FPT}
	The problem \abdeq{\csl}{|M|} is \FPT if $\clos\csl \subseteq\IV_2$.
\end{corollary}
\begin{proof}
	Immediate due to Lemma~\ref{monotone} in combination with Lemma~\ref{lem:ABD-M-IM-FPT}.
\end{proof}
\begin{lemma}\label{lem:ABD-eq-M-IS1h2-hardness}
	For any constraint language $\csl$ such that $\neg x\lor \neg y \in \closneq\csl$, the problem $\abdeq{\csl}{|E|}$ is $\para\NP$-hard.% 
\end{lemma}
\begin{proof}
We prove that the $1$-slice of the problem is $\NP$-complete by reducing from classical $\IndSet$ (which is $\NP$-complete \cite{DBLP:conf/coco/Karp72}) to $\abdeq{\csl}{}$.
The reduction is essentially the classical counterpart of the one presented in Lemma~\ref{lem:ABD-eq-E-IS1h2-hardness}.
Let $\langle V, \tilde E\rangle$ be an instance of Independent-Set. 
We map it to $\abdformat{V,H,M,\KB,s}$, where
\begin{align*}
\KB & := \{(\neg x \lor \neg y) \mid (x,y) \in \tilde E\},\\
H & := \var(\KB) \cup \{z\},\\
M & := z,\\
s & := k+1.
\end{align*}
Then $(V, \tilde E)$ admits an independent set of size $k$ if and only if $\abdformat{V,H,M,\KB,s}$ admits an explanation of size $s$.
\end{proof}
\end{document}